\documentclass[11pt,preprint]{article}

\usepackage{microtype}
\usepackage[T1]{fontenc}
\usepackage[utf8]{inputenc}
\usepackage{fullpage}
\usepackage[margin=1in]{geometry}
\usepackage{amsthm,amssymb,amsmath}  
\usepackage{xspace,enumerate}
\usepackage[dvipsnames]{xcolor}
\usepackage[colorlinks=true,urlcolor=Blue,citecolor=Green,linkcolor=BrickRed]{hyperref}
\usepackage{thmtools}
\usepackage{thm-restate}
\usepackage{authblk}
\usepackage{algorithmic}
\usepackage{todonotes}
\usepackage[noline,noend,linesnumbered,algo2e,ruled]{algorithm2e} 
\usepackage[noadjust]{cite}
\usepackage{subcaption}
 
\theoremstyle{plain}
\newtheorem{theorem}{Theorem}
\newtheorem{lemma}[theorem]{Lemma}  
\newtheorem{corollary}[theorem]{Corollary}  

\newtheorem{definition}[theorem]{Definition}
\newtheorem{example}{Example}
\title{On Two Measures of Distance between Fully-Labelled Trees}
\author[1]{Giulia Bernardini}
\author[1]{Paola Bonizzoni}
\author[2]{Paweł Gawrychowski}
\affil[1]{DISCo, Universit\`{a} degli Studi Milano - Bicocca, Italy}
\affil[2]{Institute of Computer Science, University of Wrocław, Poland}

\newcommand{\cO}{\mathcal{O}}

\newcommand{\cG}{\mathcal{G}}
\newcommand{\cM}{\mathcal{M}}
\newcommand{\cI}{\mathcal{I}}
\newcommand{\Tone}{T_1}
\newcommand{\Ttwo}{T_2}
\newcommand{\Fone}{F_1}
\newcommand{\Ftwo}{F_2}
\newcommand{\ctO}{\tilde{\mathcal{O}}}

\newcommand{\alg}{\textsc{ALG}}

\newcommand{\conserved}{\mathsf{conserved}}
\newcommand{\children}{\mathsf{children}}
\newcommand{\head}{\mathsf{head}}
\newcommand{\level}{\mathsf{level}}
\newcommand{\access}{\mathsf{access}}
\newcommand{\mode}{\mathsf{mode}}
\newcommand{\freq}{\mathsf{freq}}
\newcommand{\rep}{\mathsf{rep}}

\SetKwInput{Algorithm}{Algorithm}
\SetEndCharOfAlgoLine{}
\begin{document}
\date{}
\maketitle

\begin{abstract}
The last decade brought a significant increase in the amount of data and a variety of new inference methods for reconstructing the detailed
evolutionary history of various cancers. This brings the need of designing efficient procedures for comparing rooted
trees representing the evolution of mutations in tumor phylogenies. Bernardini et al. [CPM 2019] recently introduced
a notion of the rearrangement distance for fully-labelled trees motivated by this necessity. This notion originates from
two operations: one that permutes the labels of the nodes, the other that affects the topology of the tree. 
Each operation alone defines a distance that can be computed in polynomial time, while the actual rearrangement distance,
that combines the two, was proven to be NP-hard.

We answer two open question left unanswered by the previous work. First, what is the complexity of computing the permutation distance? 
Second, is there a constant-factor approximation algorithm for estimating the rearrangement distance between two arbitrary trees?
We answer the first one by showing, via a two-way reduction, that calculating the permutation distance between two trees on $n$
nodes is equivalent, up to polylogarithmic factors, to finding the largest cardinality matching in a sparse bipartite graph.
In particular, by plugging in the algorithm of Liu and Sidford Liu and Sidford [ArXiv 2020], we obtain an $\ctO(n^{4/3+o(1})$ time algorithm for computing
the permutation distance between two trees on $n$ nodes.
Then we answer the second question positively, and design a linear-time constant-factor approximation algorithm that does not
need any assumption on the trees.

\end{abstract}

\section{Introduction}
Phylogenetic trees represent a plausible evolutionary relationship between the most disparate objects: natural languages in linguistics~\cite{gray2009language,walker2012cultural,nakhleh2005comparison}, ancient manuscripts in archaeology~\cite{buneman1971recovery}, genes and species in biology~\cite{huber2006phylogenetic,huson2006application}.
The leaves of such trees are labelled by the entities they represent, while the internal nodes are unlabelled and stand for unknown or extinct items.
A great wealth of methods to infer phylogenies have been developed over the decades~\cite{felsenstein2004inferring,steel2016phylogeny}, together with various techniques to compare the output of different algorithms, e.g., by building a consensus tree that captures the similarity between a set of conflicting trees~\cite{bryant2003classification,jansson2016improved,jansson2016algorithms,DBLP:conf/icalp/GawrychowskiLSW18} or by defining a metric between two trees~\cite{dobson1975comparing,brodal2013efficient,estabrook1985comparison,dudek2019computing,robinson1979comparison,robinson1981comparison}.

Fully-labelled trees, in opposition to classical phylogenies, may model an evolutionary history where the internal nodes, just like the leaves, correspond to extant entities.
An important phenomenon that fits this model well is cancer progression~\cite{hajirasouliha2014combinatorial,nowell1976clonal}.
With the increasing amount of data and algorithms becoming available for inferring cancer evolution~\cite{malikic2019phiscs,jiao2014inferring,yuan2015bitphylogeny,bonizzoni2018does,bonizzoni2017beyond}, there is a pressing need of methods to provide a meaningful comparison among the trees produced by different approaches.
Besides the well-studied edit distance for fully-labelled trees~\cite{tai1979tree,zhang1989simple,pawlik2015efficient,mcvicar2016sumoted}, a few recent papers proposed ad-hoc metrics for tumor phylogenies~\cite{karpov2019multi,govek2018consensus,DiNardo2019,Ciccolella2020.04.14.040550}.
Taking inspiration from the existing literature~\cite{dasgupta1997distances, bordewich2005computational, allen2001subtree,steel2016phylogeny} on phylogeny rearrangement, the study of an operational notion of distance for
rearranging a fully-labelled tree is of great interest, and there are still many unexplored questions to be answered.

Following this line of research, we revisit the two notions of operational distance between fully-labelled trees recently introduced by Bernardini et al.~\cite{bernardini2019distance}. 
We consider rooted trees on $n$ nodes labelled with distinct labels from $[n]=\{1,2,\ldots,n\}$, and identify nodes with their labels.
We recall the following two basic operations on such trees:
\begin{itemize}
\item \textbf{link-and-cut operation}: given $u$, $v$ and $w$ such that $v$ is a child of $u$ and $w$ is not a descendant of $v$, the link-and-cut operation $v\,|\,u\rightarrow w$ consists of two suboperations: cut the edge $(v,u)$ and add the edge $(v,w)$, effectively switching the parent of $v$ from $u$ to $w$.
\item \textbf{permutation operation}: apply some permutation $\pi : [n] \rightarrow [n]$ to the nodes. 
If a node $u$ was a child of $v$ before the operation, then after the operation $\pi(u)$ is a child of $\pi(v)$.
\end{itemize}
The size $|\pi|$ of a permutation is the number of elements $x$ s.t. $\pi(x)\neq x$.
Two trees $\Tone$ and $\Ttwo$ are isomorphic if and only if one can reorder the children of every node so as to make the
trees identical after disregarding the labels. The \emph{permutation distance} $d_{\pi}(T_1, T_2)$ between two isomorphic trees is the smallest size $|\pi|$ of a permutation $\pi$ that transforms $\Tone$ into $\Ttwo$.
Bernardini et al.~\cite{bernardini2019distance} designed a cubic time algorithm for computing the permutation distance.

The size of a sequence of link-and-cut and permutation operations is the sum of the number of link-and-cut
operations and the total size of all permutations.
The \emph{rearrangement distance} $d(\Tone,\Ttwo)$ between two (not necessarily isomorphic) trees with identical roots is the smallest size of any sequence of link-and-cut and permutation operations that, without permuting the root, transform $\Tone$
into $\Ttwo$.
Bernardini et al.~\cite{bernardini2019distance} proved that computing the rearrangement distance is NP-hard,
but for binary trees there exists a polynomial time 4-approximation algorithm.

We consider two natural open questions. First, what is the complexity of
computing the permutation distance? Second, is there a constant-factor approximation algorithm for estimating
the rearrangement distance between two arbitrary trees?
For computing the permutation distance, in Section~\ref{sec:permutation} we connect the complexity to that of calculating the largest cardinality matching in a sparse bipartite graph. By designing two-way reductions we show that these problems are equivalent, up to polylogarithmic factors.
Due to the recent progress in the area of fine-grained complexity we now know, for many problems that can be solved in polynomial
time, what is essentially the best possible exponent in the running time, conditioned on some plausible but yet unproven
hypothesis~\cite{Williams18}. For max-flow, and more specifically maximum matching, this is not the case yet, although
we do have some understanding of the complexity of the related problem of computing the max-flow between all pairs of
nodes~\cite{AbboudKT20,KrauthgamerT18,AbboudGIKPTUW19}. So, even though our reductions don't tell us what is the best possible
exponent in the running time, they do imply that it is the same as for maximum matching in a sparse bipartite graph.
In particular, by plugging in the asymptotically fastest known algorithm~\cite{liu2020faster}, we obtain an $\ctO(n^{4/3+o(1)})$ time algorithm for computing
the permutation distance between two trees on $n$ nodes. The main technical novelty in our reduction from permutation
distance is that, even though the natural approach would result in multiple instances of weighted maximum bipartite matching,
we manage to keep the graphs unweighted.

For the rearrangement distance, in Section~\ref{sec:approx} we design a linear-time constant-factor approximation algorithm that does
not assume that the trees are binary. The algorithm consists of multiple phases, each of them introducing more and
more structure into the currently considered instance, while making sure that we don't pay more than the optimal distance
times some constant. To connect the number of steps used in every phase with the optimal distance, we introduce a new
combinatorial object that can be used to lower bound the latter inspired by the well-known algorithm
for computing the majority~\cite{Moore91}.

\section{Preliminaries}

Let $[n]=\{1,2,\ldots,n\}$.
We consider rooted trees and forests on nodes labelled with distinct labels from $[n]$, and identify nodes with their
labels.
The parent of $u$ in $F$ is denoted $p_{F}(u)$, and we use the convention
that $p_{F}(u)=\bot$ when $u$ is a root in $F$. $F|u$ denotes the subtree of $F$ rooted at $u$,
$\children_{F}(u)$ stands for the set of children of a node $u$ in $F$, and $\level_{F}(u)$ is the level of $u$ in $F$
(with the roots being on level 0).

Two trees $\Tone$ and $\Ttwo$ are isomorphic, denoted $\Tone \equiv \Ttwo$, if and only if there exists a bijection $\mu$ between their nodes
such that, for every $u\in[n]$ with $p_{T_{1}}(u)\neq\bot$, it holds that $\mu(p_{T_{1}}(u))=p_{T_{2}}(\mu(u))$, implying in particular that $\mu$ maps the root of $\Tone$ to the root of $\Ttwo$.
Let $\cI(\Tone,\Ttwo)$ denote the set of all such bijections $\mu$.
Given two isomorphic trees $\Tone$ and $\Ttwo$, we seek a permutation $\pi$ with the smallest size that
transforms $\Tone$ into $\Ttwo$. This is equivalent to finding $\mu\in\cI(\Tone,\Ttwo)$ that maximises
the number of conserved nodes $\conserved(\mu)=\{ u :  u = \mu(u) \}$,
as these two values sum up to $n$. 

When working on the rearrangement distance, for ease of presentation, instead of the link-and-cut operation we will work with the cut operation defined as follows:
\begin{itemize}
\item \textbf{cut operation}: let $u,v$ be two nodes such that $v$ is a child of $u$. The cut operation
$(v\dagger u)$ removes the edge $(v,u)$, effectively making $v$ a root.
\end{itemize}
The size of a sequence of cut and permutation operations is defined similarly as for a sequence
of link-and-cut and permutation operations.
Since a permutation operation is essentially just renaming the nodes, we can assume
that all permutation operations precede all link-and-cut (or cut) operations, or vice versa.
Furthermore, multiple consecutive permutation operations can be replaced by a single permutation operation
without increasing the total size.

This leads to the notion of rearrangement distance between two forests $\Fone$ and $\Ftwo$.
We write $\Fone\sim \Ftwo$ to denote that, for every $u\in [n]$, at least one of the following three conditions holds:
(i) $p_{\Fone}(u) = p_{\Ftwo}(u)$,
(ii) $p_{\Fone}(u)=\bot$, or
(iii) $p_{\Ftwo}(u)=\bot$.
The rearrangement distance $\tilde{d}(\Fone,\Ftwo)$ is the
smallest size of any sequence of cut and permutation operations that transforms $\Fone$ into $\Fone'$ such that $\Fone' \sim \Ftwo$.
This is the same as the smallest size of any sequence of cut and permutation operations
that transforms $\Ftwo$ into $\Ftwo'$ such that $\Fone \sim \Ftwo'$, as both sizes are equal
to the minimum over all permutations $\pi$ that fix the original root of the following expression
\[ |\{ u : \pi(u) \neq u\}| + |\{ u : p_{\Fone}(u) \neq p_{\Ftwo}(\pi(u))~\land~p_{\Fone}(u)\neq \bot~\land~p_{\Ftwo}(\pi(u))\neq \bot \}| . \]
Consequently, $\tilde{d}$ defines a metric.
The original notion of rearrangement distance $d$ between two trees was similarly defined as the smallest size of any sequence of link-and-cut and permutation operations that transforms $\Tone$ into $\Ttwo$, under the additional assumption that the roots of
$\Tone$ and $\Ttwo$ are identical (so $d(\Tone,\Ttwo)$ is well-defined) and cannot participate in any
permutation operation~\cite{bernardini2019distance}.
In Section~\ref{sec:approx} we connect $d(\Tone,\Ttwo)$ and $\tilde{d}(\Tone,\Ttwo)$, and then work with the latter.

A matching in a bipartite graph is a subset of edges with no two edges meeting at the same vertex. A maximum matching
in an unweighted bipartite graph is a matching of maximum cardinality, whereas a maximum weight matching in a weighted
bipartite graph is a matching in which the sum of weights is maximised. 
Given an unweighted bipartite graph with $m$ edges, the well-known algorithm by Hopcroft and Karp~\cite{HopcroftK73} finds
a maximum matching in $\cO(m^{1.5})$ time. This has been recently improved by Liu and Sidford to $\tilde\cO(m^{4/3+o(1)})~$\cite{liu2020faster}. 

A \emph{heavy path decomposition} of a tree $T$ is obtained by selecting, for every non-leaf node $u\in T$, its \emph{heavy child} $v$ such
that $T|v$ is the largest: there will be some subtlety in how to resolve a tie in this definition that will be explained in detail later.
This procedure decomposes the nodes of $T$ into node-disjoint paths called \emph{heavy paths}.
Each heavy path $p$ starts at some node, called its \emph{head}, and ends at a leaf: $\head_{T}(u)$ denotes the head of the heavy path containing a node $u$ in $T$. An important property of such a decomposition is
that the number of distinct heavy paths above any leaf (that is, intersecting the path from a leaf to the root) is only logarithmic
in the size of $T$~\cite{SleatorT83}. 

\section{A Fast Algorithm for the Permutation Distance}
\label{sec:permutation}
Our aim is to find $\mu\in\cI(\Tone,\Ttwo)$ that maximises $\conserved(\mu)$, that is
$\gamma(\Tone,\Ttwo)=\max \{ \conserved(\mu) : \mu \in \cI(\Tone,\Ttwo)\}$. To make the notation less
cluttered, we define $\gamma(x,y)=\gamma(\Tone|x,\Ttwo|y)$.
Let us start by describing a simple polynomial time algorithm which follows the construction of~\cite{bernardini2019distance}.
We will then show how to improve it to obtain a faster algorithm that uses unweighted bipartite maximum matching.
Finally, we will show a reduction from bipartite maximum matching to computing
the permutation distance, establishing that these two problems are in fact equivalent, up to polylogarithmic factors.

\subsection{Polynomial Time Algorithm }\label{subsec:n2}

We first run the folklore linear-time algorithm of~\cite{Aho} for determining if two rooted trees are isomorphic.
Recall that this algorithm assigns a number from $\{1,2,\ldots,2n\}$ to every node of $\Tone$ and $\Ttwo$ so that the subtrees
rooted at two nodes are isomorphic if and only if their numbers are equal.
The high-level idea is then to consider a weighted bipartite graph
$G(u,v)$ for each $u,v\in [n]$ such that $\level_{\Tone}(u)=\level_{\Ttwo}(v)$ and $\Tone|u \equiv \Ttwo|v$.
The vertices of $G(u,v)$ are $\children_{\Tone}(u)$ and $\children_{\Ttwo}(v)$, and there is
an edge of weight $\gamma(u',v')$ between $u' \in \children_{\Tone}(u)$ and $v' \in \children_{\Ttwo}(v)$ if and only if
$\Tone|u' \equiv \Ttwo|v'$ and $\gamma(u',v')>0$. 
We call such graphs the \emph{distance graphs} for $\Tone$ and $\Ttwo$ and denote them collectively by $\cG(\Tone,\Ttwo)$.

$\gamma(u,v)$ is computed as follows, with $\cM(G(u,v))$ denoting the weight of
a (not necessarily perfect) maximum weight matching in $G(u,v)$, $\Gamma(u,v)=1$ if $u=v$ and $\Gamma(u,v)=0$ otherwise.
\begin{equation}
\label{eq:recursion}
  \gamma(u,v) = \left\{
    \begin{array}{ll}
      \cM(G(u,v))+ \Gamma(u,v) & \textrm{ if } T_1|u \equiv T_2|v,\\
      0 & \textrm{ otherwise.}\\
    \end{array}
  \right.
\end{equation}
The overall number of edges created in all graphs is $\cO(n^2)$. Indeed, for each $u\in [n]$
such that $\level_{\Tone}(u)=\level_{\Ttwo}(u)$ and $\Tone|u \equiv \Ttwo|u$,
and for each pair of ancestors $z$ of $u$ in $\Tone$ and $w$ of $u$ in $\Ttwo$ such that $\level_{\Tone}(z)=\level_{\Ttwo}(w)$ and
$\Tone|z \equiv \Ttwo|w$, we possibly add an edge $(z,w)$ to the graph $G(p_{\Tone}(z),p_{\Ttwo}(w))$.
Since there are up to $n$ 
pairs of ancestors on the same level for each label, and the labels are $n$, there are $\cO(n^2)$ edges overall.

We then start from the deepest level in both trees, and we move up level by level
towards the roots in both trees simultaneously. For each level $k$, we consider all pairs of isomorphic subtrees rooted at level $k$,
build the corresponding distance graphs, and use Equation (\ref{eq:recursion}) to weigh the edges. After having reached the roots, we
return the value of $\gamma(\Tone,\Ttwo)$. The correctness of the algorithm follows directly from Lemma 13 of~\cite{bernardini2019distance}, stating that the permutation distance is equal to the minimum number of labels that are not conserved by any isomorphic mapping, i.e., $d_{\pi}(\Tone,\Ttwo)=n-\gamma(\Tone,\Ttwo)$. The running time is polynomial if we plug in any polynomial-time maximum weight matching algorithm.

In the next subsection we show how to obtain a better running time by constructing a different version of distance graphs, so that the total weight of their edges will be subquadratic,
and replacing maximum weight matching with maximum matching.

\subsection{Reduction to Bipartite Maximum Matching}
\label{sub:reduction}
We start by finding a heavy path decomposition of $\Tone$ and $\Ttwo$, with some extra care in resolving a tie if there are multiple children with subtrees of the same size, as follows.
Recall that we already know which subtrees of $\Tone$ and $\Ttwo$ are isomorphic, as the algorithm of~\cite{Aho} assigns the same number from $\{1,2,\ldots,2n\}$ to nodes of $\Tone$ and $\Ttwo$ with isomorphic subtrees.
For every $u,v\in [n]$ such that $\Tone|u \equiv \Ttwo|v$, we would like the heavy child $u'$ of $u$ in $\Tone$
and $v'$ of $v$ in $\Ttwo$ to be such that $\Tone|u' \equiv \Ttwo|v'$. 
This can be implemented in linear time: it suffices to group the
nodes with isomorphic subtrees together, and then make the choice just once for every such group.

Consider a graph $G(u,v)$ for some $u,v\in [n]$: the edge corresponding to the heavy child $u'$ of $u$ in $\Tone$
and the heavy child $v'$ of $v$ in $\Ttwo$ is called \emph{special} (note that this edge might not exist).
The key observation is that the properties of heavy path decomposition
allow us to bound the total weight of non-special edges by $\cO(n\log n)$.

\begin{figure}[t]
    \centering
   \includegraphics[width=0.73\linewidth]{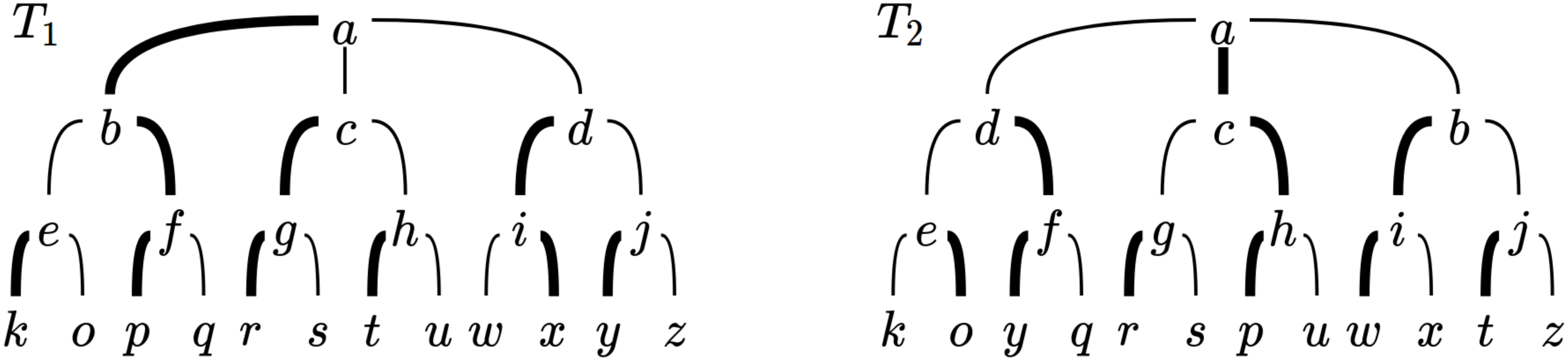}
    \caption{$\Tone$ and $\Ttwo$ with a possible heavy path decomposition.}
    \label{fig:paths}
\end{figure}

\begin{figure}
\centering
    \begin{subfigure}[b]{0.3\textwidth}
        \includegraphics[width=\linewidth]{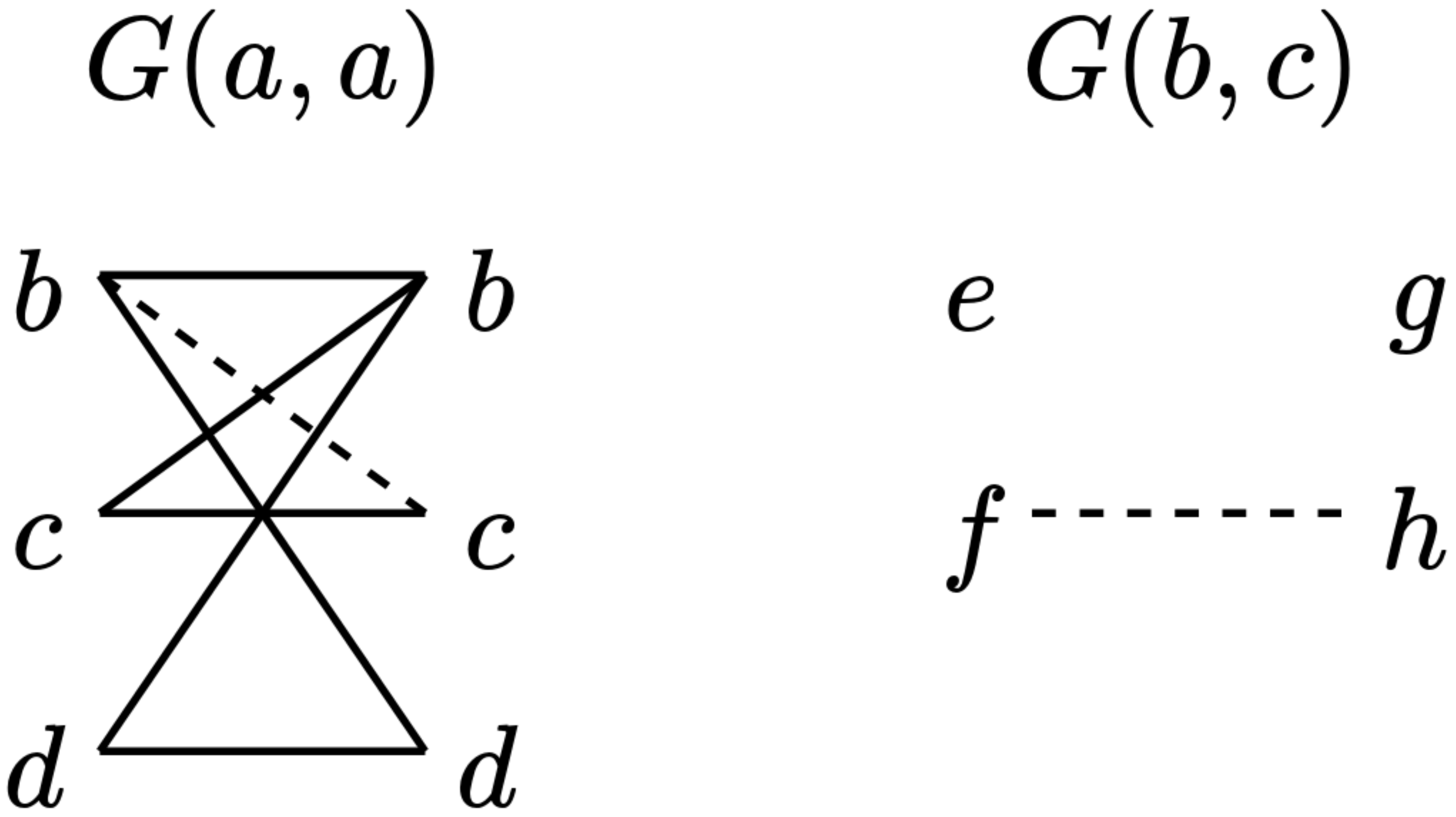}
        \vspace{-3mm}
    \end{subfigure}
    ~~~~~~~~
    \begin{subfigure}[b]{0.105\textwidth}
        \includegraphics[width=\linewidth]{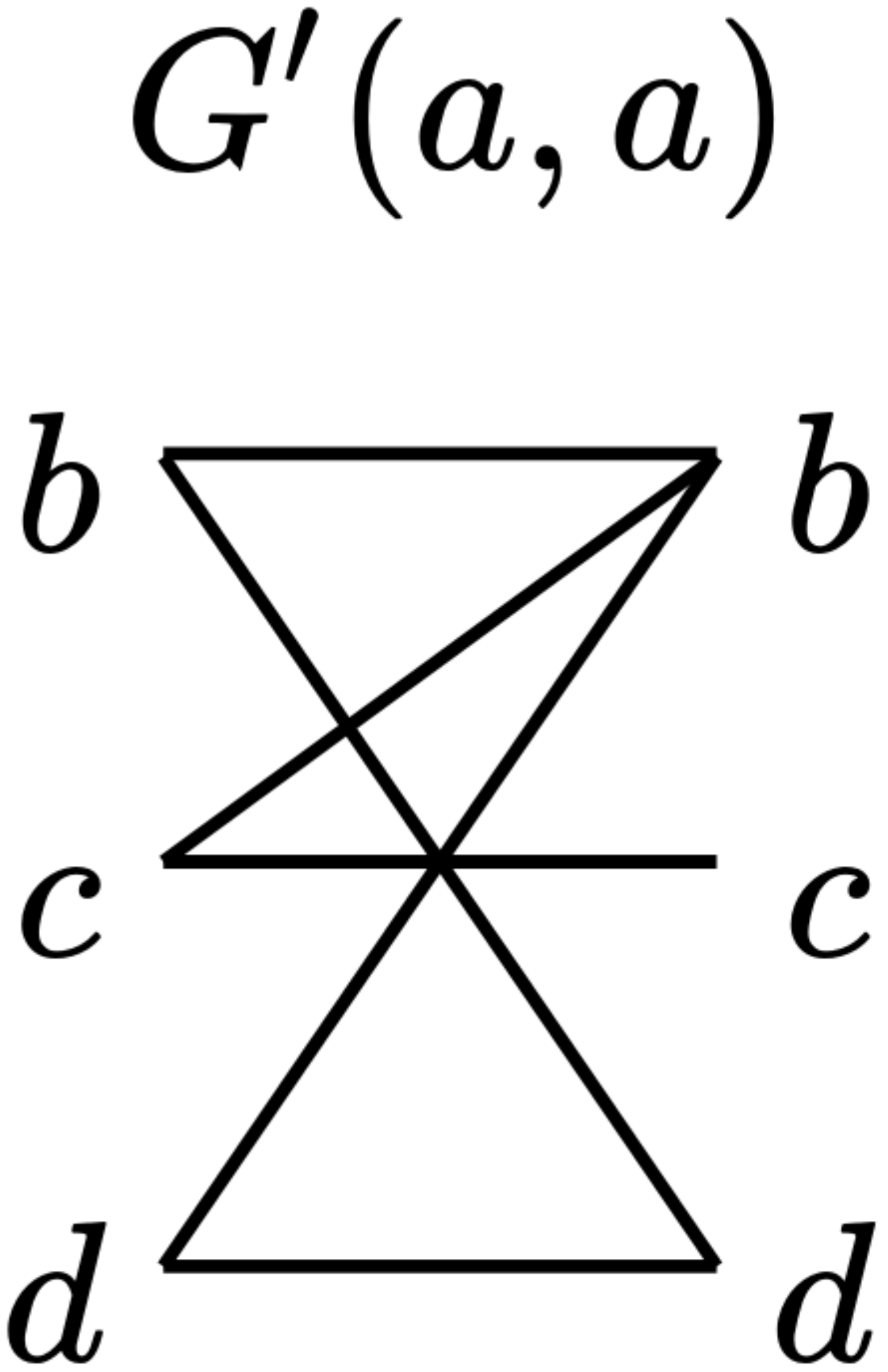}
        \vspace{-3mm}
    \end{subfigure}
     ~~~~~~~~
    \begin{subfigure}[b]{0.105\textwidth}
        \includegraphics[width=\linewidth]{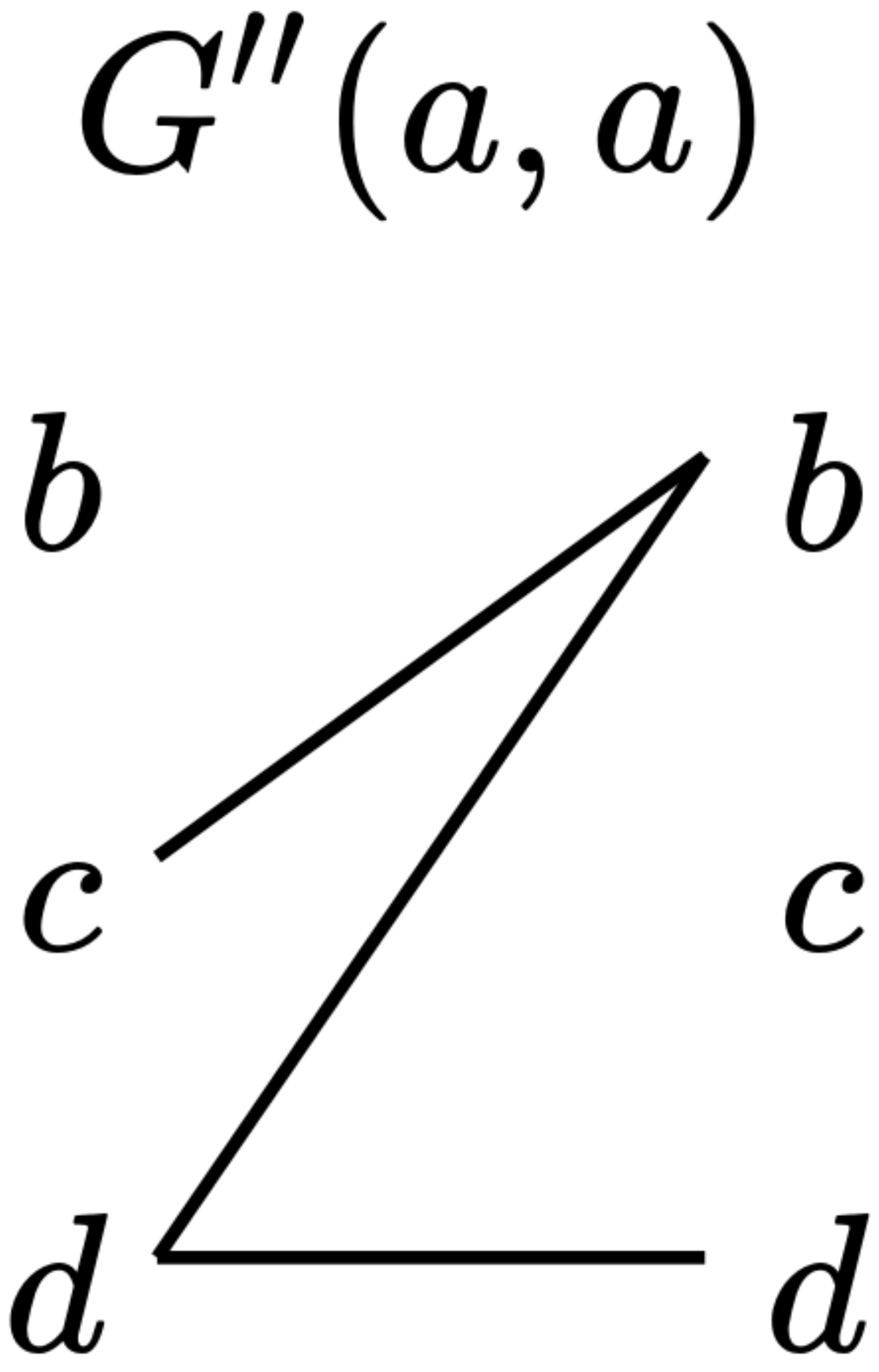}
        \vspace{-3mm}
    \end{subfigure}
    
    \caption{$G(a,a)$ (type 1), $G'(a,a)$, $G''(a,a)$ and $G(b,c)$ (type3) for the trees of Figure~\ref{fig:paths}. The special edge in each graph is dashed.}
    \label{fig:graphs}
\end{figure}

\begin{lemma}
\label{lem:nonspecial}
The total weight of all non-special edges in $\cG(\Tone,\Ttwo)$ is $\cO(n\log n)$.
\end{lemma}

\begin{proof}
Consider any $u\in [n]$ such that $\level_{\Tone}(u)=\level_{\Ttwo}(u)$ and $\Tone|u \equiv \Ttwo|u$.
For each pair of ancestors $z$ of $u$ in $\Tone$ and $w$ of $u$ in $\Ttwo$ such that $\level_{\Tone}(z)=\level_{\Ttwo}(w)$,
$\Tone|z \equiv \Ttwo|w$ and either $\head_{\Tone}(z)=z$ or $\head_{\Ttwo}(w)=w$, $u$ will contribute 1 to the weight of an edge $(z,w)$ in $G(p_{\Tone}(z),p_{\Ttwo}(w))$.
Because there are at most $\log n$ heavy paths above
any node of $\Tone$ or $\Ttwo$, each label $u\in [n]$ contributes 1 to the weight of at most $2\log n$ non-special edges, making their total weight $\cO(n\log n)$ overall.
\end{proof}

We divide the graphs in $\cG(\Tone,\Ttwo)$ into three types (see Figure~\ref{fig:graphs}, left, for an example):
\begin{description}
\item[Type 1: ] graphs $G(u,v)$ with at least one non-special edge.
\item[Type 2: ] graphs $G(u,v)$ with no non-special edges, and $\Gamma(u,v)=1$.
\item[Type 3: ] graphs $G(u,v)$ with no non-special edges, and $\Gamma(u,v)=0$.
\end{description}

We will construct only the graphs of type 1 and 2, and extract from them the information that the graphs of type 3 would have captured.
In what follows we show  how to construct the graphs of type 1 and 2 in $\cO(n\log^{2} n)$ time.

\paragraph{Constructing the Graphs of Type 1 and 2.}\label{subsub:constructing}
The first step is to find all pairs of nodes that correspond to graphs of type 1 or 2, and store them
in a dictionary $D$ implemented as a balanced search tree with $\cO(\log n)$ access time.
The second step is to find the non-special edges of these graphs, and store them in a separate dictionary, also implemented as a balanced search tree with $\cO(\log n)$ access time.
Note that the weights will be found at a later stage of the algorithm.
We assume that both trees have been already decomposed into heavy paths, and we already know which subtrees are isomorphic. 
This can be preprocessed in $\cO(n)$ time. 

\begin{lemma}
All graphs of type 1 and 2 can be identified in $\cO(n\log^2 n)$ time.
\end{lemma}

\begin{proof}
We consider every $u\in [n]$ such that $\level_{\Tone}(u)=\level_{\Ttwo}(u)$
and $\Tone | u \equiv \Ttwo | u$ in two passes. In the first pass,
we need to iterate over every ancestor $z$ of $u$ in $\Tone$
and $w$ of $u$ in $\Ttwo$ such that $\level_{\Tone}(z)=\level_{\Ttwo}(w)$, $T_{1}|z \equiv T_{2}|w$ and either $\head_{\Tone}(z)=z$ or $\head_{\Ttwo}(w)=w$, 
and if additionally $T_{1}|p_{T_{1}}(z) \equiv T_{2}|p_{T_{2}}(w)$ then designate $G(p_{T_{1}}(z),p_{T_{2}}(w))$ to be a graph of type 1 and insert it into $D$. 
As a non-special edge $(z,w)$ of a graph $G(p_{T_{1}}(z),p_{T_{2}}(w))$ is such that either $z$ or $w$ are not on the same heavy path as their parents, this correctly determines all graphs of type 1.

To efficiently iterate over all such $z$ and $w$ given $u$,
we assume that the nodes of every heavy path of a tree $T$ are stored in an array, so that, given any node $u\in T$, we are able to access the node that belongs to the same heavy path as $u$ and whose level is $\ell$ in constant time, if it exists. 
We denote such operation $\access_{T}(u,\ell)$.
Given two nodes $u\in\Tone$ and $v\in\Ttwo$ on the same level, the procedure below shows how to iterate over every ancestor $z$ of $u$ and $w$ of $v$ such that $\level_{\Tone}(z)=\level_{\Ttwo}(w)$ and either $\head_{\Tone}(z)=z$ or $\head_{\Ttwo}(w)=w$, in $\cO(\log n)$ time,
implying that all graphs of type 1 can be identified in $\cO(n\log^{2}n)$ time.

{
\setlength{\interspacetitleruled}{0pt}%
\setlength{\algotitleheightrule}{0pt}%
\begin{algorithm2e}[H]
\While{$u \neq \bot$ {\bf and} $v\neq \bot$}{
  \uIf{$ \level_{T_{1}}(\head_{T_{1}}(u)) < \level_{T_{2}}(\head_{T_{2}}(v)) $}{
    output $\access_{T_{1}}(u,\level_{T_{2}}(\head_{T_{2}}(v)))$ and $\head_{T_{2}}(v)$\;
    $v \gets p_{T_{2}}(\head_{T_{2}}(v))$\;
  }\eIf{$ \level_{T_{1}}(\head_{T_{1}}(u)) > \level_{T_{2}}(\head_{T_{2}}(v)) $}{
    output $\head_{T_{1}}(u)$ and $\access_{T_{2}}(v,\level_{T_{1}}(\head_{T_{1}}(u)))$\;
    $u \gets p_{T_{1}}(\head_{T_{1}}(u))$\;
  }{
    output $\head_{T_{1}}(u)$ and $\head_{T_{2}}(v)$\;
    $u \gets p_{T_{1}}(\head_{T_{1}}(u))$\;
    $v \gets p_{T_{2}}(\head_{T_{2}}(v))$\;
  }
}
\end{algorithm2e}
}

In the second pass,
for each $u\in [n]$ such that $\level_{\Tone}(u)=\level_{\Ttwo}(u)$
and $\Tone | u \equiv \Ttwo | u$, we designate $G(u,u)$ to be a graph of type 2, unless it has been already designated to be a graph of type 1.
\end{proof}

\begin{lemma}
All graphs of type 1 and 2 can be populated with their edges in $\cO(n\log^2 n)$ time.
\end{lemma}

\begin{proof}
For each such graph $G(u,v)$ such that none of $u,v$
is a leaf, let $u'$ be the unique heavy child of $u$, and $v'$ be the unique heavy child of $v$. We add the special
edge $(u',v')$ to $G(u,v)$. To find the non-special edges, we again consider every $u\in [n]$ such that
$\level_{\Tone}(u)=\level_{\Ttwo}(u)$ and $\Tone|u \equiv \Ttwo|u$: we iterate over the ancestors $z$ of $u$ in $\Tone$
and $w$ of $u$ in $\Ttwo$ such that $\level_{\Tone}(z)=\level_{\Ttwo}(w)$, $T_{1}|z \equiv T_{2}|w$ and either $\head_{\Tone}(z)=z$ or $\head_{\Ttwo}(w)=w$, and if additionally $T_{1}|p_{\Tone}(z) \equiv T_{2}|p_{\Ttwo}(w)$ then add a non-special edge
$(z,w)$ to $G(p_{\Tone}(z),p_{\Ttwo}(w))$ .
This takes $\cO(n\log^{2}n)$ time overall.
\end{proof}

\paragraph{Processing the Graphs of Type 1 and 2.}\label{subsub:processing}
Having constructed the graphs of type 1 and 2 in $\cO(n\log^{2}n)$ time, we process them level by level.
Consider $G(u,v)$: for each of its edges $(u',v')$ corresponding to $u'\in\children_{T_{1}}(u)$
and $v'\in\children_{T_{2}}(v)$, we need to extract its weight $\gamma(u',v')$. If $G(u',v')$ is of type 1 or 2,
the graph
can be extracted from the dictionary in $\cO(\log n)$ time. Otherwise,
$G(u',v')$ is of type 3 and we need to make up for not having processed such graphs.

To this aim, we associate a sorted list of levels with each pair of heavy paths of $\Tone$ and $\Ttwo$.
The lists are stored in a dictionary indexed by the heads of the heavy paths.
For every $u,v\in[n]$ such that $G(u,v)$ is of type 1 or 2, we append
the levels of $u$ and $v$ to the lists associated with the respective heavy paths.
The lists can be constructed in $\cO(n\log^{2}n)$ time by processing the graphs level by level, and allow us to efficiently use the following lemma.

\begin{lemma}
\label{lem:type3}
Consider $u,v\in [n]$ such that $\level_{\Tone}(u)=\level_{\Ttwo}(v)$ and $\Tone|u\equiv \Ttwo|v$,
but $G(u,v)$ is of type 3. Either both $u$ and $v$ are leaves and $\gamma(u,v)=0$,
or the heavy child of $u$ is $u'$, the heavy child of $v$ is $v'$, and $\gamma(u,v)=\gamma(u',v')$.
\end{lemma}

\begin{proof}
First observe that $u\neq v$, as otherwise $G(u,v)$ would be of type 2.
Becase $\Tone|u\equiv \Ttwo|v$, either both $u$ and $v$ are leaves or none of them is a leaf.
In the former case, $G(u,v)$ is empty and $\gamma(u,v)=0$.
By how we resolve ties in the heavy path decomposition, in the latter case we have $\Tone|u' \equiv \Ttwo|v'$,
where $u'$ is the heavy child of $u$ and $v'$ is the heavy child of $v$.
$G(u,v)$ consists of the unique special edge corresponding to the heavy child $u'$ of $u$
and $v'$ of $v$, so $\cM(G(u,v))$ is equal to the
cost of the special edge, and by (\ref{eq:recursion}) we obtain that $\gamma(u,v)=\gamma(u',v')$.
\end{proof}

Given $u,v\in [n]$ such that $\level_{\Tone}(u)=\level_{\Ttwo}(v)=\ell$ and $\Tone|u\equiv \Ttwo|v$,
we extract $\gamma(u,v)$ by accessing the sorted list associated with the heavy paths of $u$ and $v$: we binary search for the smallest level $\ell'\geq \ell$ such that the heavy paths of $u$ and $v$ respectively contain a node
$u'$ and $v'$, both on level $\ell'$, with $G(u',v')$ of type
1 or 2. Then Lemma~\ref{lem:type3}, together with the fact that in our heavy path decomposition the subtrees rooted at
the heavy children of two nodes with isomorphic subtrees are also isomorphic, implies that
$\gamma(u,v)=\gamma(u',v')$.

It remains to describe how to compute $\cM(G(u,v))$ for every graph $G(u,v)$ of type 1 and 2. We could
have used any maximum weight matching algorithm, but this would result in a higher running time.
Our goal is to plug in a maximum matching algorithm. This seems problematic as $G(u,v)$ is a weighted
bipartite graph, but we will show that maximum weight matching can be reduced to multiple instances
of maximum matching. However, bounding the overall running time will require bounding the total weight
of all edges belonging to graphs of type 1 and 2. By Lemma~\ref{lem:nonspecial} we already know that the total
weight of all non-special edges is $\cO(n\log n)$, but such bound doesn't hold for the special edges.
Therefore, we proceed as follows. Let $u'$ be the heavy child of $u$ and $v'$ be the heavy child of $v$.
We construct $G'(u,v)$ by removing the special edge from $G(u,v)$.
We also construct $G''(u,v)$ from $G(u,v)$ by removing all the edges incident to $u'$ and $v'$ (see Figure~\ref{fig:graphs} for an example). 
Equation~(\ref{eq:recursion}) can then be rewritten as follows:
\begin{equation}
\label{eq:weight}
    \gamma(u,v)=\max\{\cM(G'(u,v)), \cM(G''(u,v))+\gamma(u',v') \} + \Gamma(u,v) .
\end{equation}
This is because a maximum weight matching in $G(u,v)$ either includes the special edge $(u',v')$, implying that no other edges incident to $u'$ and $v'$ can be part of the matching and thus $\cM(G(u,v))=\cM(G''(u,v))+\gamma(u',v')$, or it does not include it, thus $\cM(G(u,v))=\cM(G'(u,v))$.
Since the graphs $G'(u,v)$ and $G''(u,v)$ contain only non-special edges,
the overall weight of all edges in the obtained instances of maximum weight matching is $\cO(n\log n)$.

We already know that constructing all the relevant graphs takes $\cO(n\log^{2}n)$ time. It remains to
analyze the time to calculate the maximum weight matching in every $G'(u,v)$ and $G''(u,v)$.
We first present a preliminary lemma that connects the complexity of calculating the maximum weight matching
in a weighted bipartite graph to the complexity of calculating the maximum matching in an unweighted bipartite graph.

\begin{lemma}[\cite{decomposition}]
\label{lem:decomposition}
Let $G$ be a weighted bipartite graph, and let $N$ be the total weight of all the edges of $G$. Calculating the
maximum weight matching in $G$ can be reduced in $\cO(N)$ time to multiple instances of calculating the
maximum matching in an unweighted bipartite graph, in such a way that the total number of edges in all such graphs is at most $N$.
\end{lemma}
\begin{proof}
Using the decomposition theorem of Kao, Lam, Sung, and Ting~\cite{decomposition}, we can reduce computing
the maximum weight matching in a weighted bipartite graph such that the total weight of all edges is $N$ to
multiple instances of calculating the largest cardinality matching in an unweighted bipartite graph. The total
number of edges in all unweighted bipartite graphs is $\sum_{i}m_{i}=N$ and the reduction can be implemented
in $\cO(N)$ time by maintaining a list of edges with weight $w$, for every $w=1,2,\ldots,N$.
\end{proof}

\begin{theorem}
\label{thm:tomatching}
Let $f(m)$ be the complexity of calculating the maximum matching in an unweighted bipartite graph on $m$ edges,
and let $f(m)/m$ be nondecreasing.
The permutation distance can be computed in $\tilde{\cO}(f(n))$ time.
\end{theorem}

\begin{proof}
The total number of edges in all constructed graphs is $\cO(n\log n)$, and the total time to construct the relevant
graphs and extract the costs of their edges is $\cO(n\log^{2}n)$.
Thus, the total running time is $\cO(n\log^{2}n)$ plus the time to compute the maximum weight matching in every graph of type 1
and type 2. Let $N_{i}$ be the total weight of all non-special edges in the $i$-th of these graphs.
By Lemma~\ref{lem:nonspecial}, $\sum_{i} N_{i} = \cO(n\log n)$.
Additionally, $N_{i} \leq n$. 
Let $m_{i,j}$ be the number of edges in the $j$-th instance of unweighted bipartite matching for the $i$-th graph.
By Lemma~\ref{lem:decomposition}, the overall time is hence $\sum_{i,j} f(m_{i,j})$, where
$\sum_{i,j} m_{i,j} \leq  \sum_{i} N_{i} = \cO(n\log n)$ and $m_{i,j} \leq N_{i} \leq n$.
We upper bound $\sum_{i,j}f(m_{i,j})$ using the assumption that $f(m)/m$ is nondecreasing as follows:
\[
\sum_{i,j} f(m_{i,j}) = \sum_{i,j} m_{i,j} \cdot f(m_{i,j})/m_{i,j} \leq \sum_{i,j} m_{i,j} \cdot f(n)/n = \cO(f(n)\log n). \qedhere
\]
\end{proof}

\begin{corollary}
The permutation distance can be computed in $\tilde{\cO}(n^{4/3+o(1)})$ time.
\end{corollary}

\subsection{Reduction from Bipartite Maximum Matching}
\label{sec:from}

We complement the algorithm described in Subection~\ref{sub:reduction} with a reduction from
bipartite maximum matching to computing the permutation distance: see Figure~\ref{fig:reduction} for an example.

\begin{theorem}
\label{thm:frommatching}
Given an unweighted bipartite graph on $m$ edges, we can construct in $\cO(m)$ time two trees
with permutation distance equal to the cardinality of the maximum matching.
\end{theorem}
\begin{proof}
We first modify the graph so that the degree of every node is at most 3. This can be ensured in $\cO(m)$
time by repeating the following transformation: take a node $u$ with neighbours $v_{1},v_{2},\ldots,v_{k}$, $k\geq 4$.
Replace $u$ with $u'$ and $u''$ both connected to a new node $v$, connect $u'$ to $v_{1},v_{2},\ldots,v_{k-2}$ and
$u''$ to $v_{k-1},v_{k}$. It can be verified that the cardinality of the maximum matching in the new graph
is equal to that in the original graph increased by 1. By storing, for every node, the incident edges in a linked list, we can
implement a single step of this transformation in constant time, and there are at most $m$ steps.

We will now first construct two unlabelled trees and then explicitly assign appropriate labels to their nodes.
Without loss of generality, let the nodes of the graph be $u_{1},u_{2},\ldots,u_{m}$ and $v_{1},v_{2},\ldots,v_{m}$.
In the first tree we create $m$ nodes, labelled with $u_{1},u_{2},\ldots,u_{m}$,
connected to a common unlabelled root. In the second tree we do the same with nodes $v_{1},v_{2},\ldots,v_{m}$. Then, for every edge $(u_{i},v_{j})$ of the graph, we attach a new leaf  to $u_{i}$ in the first tree and to $v_{j}$ in the
second tree, and assign the same label to both of them. Finally, we attach enough unlabelled leaves to every $u_{i}$ and $v_{j}$ to make their degrees all equal to 3. 
To make both trees fully-labelled on the same set of labels, we further attach $1+m+3m-m=3m+1$ \emph{extra leaves} to the roots of both trees. 
For every unlabelled leaf attached to $u_{1},u_{2},\ldots,u_{m}$ of the first tree, we choose an unlabelled extra leaf of the second tree, and assign the same label to both of them.
We then assign the same label to the root of the first tree and an extra leaf of the second tree, and label the last $m$ extra leaves of the second tree with $u_{1},u_{2},\ldots,u_{m}$.
We finally swap the trees and repeat the same procedure: see Figure~\ref{fig:reduction} for an example. 

The permutation distance between the two trees is equal
to the cardinality of the maximum matching. Indeed, the trees are clearly isomorphic; moreover, any isomorphism must match extra leaves with extra leaves, and every $u_{i}$ to a $v_{\pi(j)}$, for some permutation $\pi$ on $[m]$. The extra leaves do not
contribute to the number of conserved nodes, while $u_{i}$ and $v_{\pi(j)}$ contribute 1 if and only if $(u_{i},v_{\pi(j)})$ was
an edge in the original graph. Thus, the distance is equal to the maximum over all permutations $\pi$
of the number of edges $(u_{i},v_{\pi(j)})$. This in turn is equal to the cardinality of the maximum matching in the original
graph.
\end{proof}

\begin{figure}
\centering
    \begin{subfigure}[b]{0.14\textwidth}
        \includegraphics[width=\linewidth]{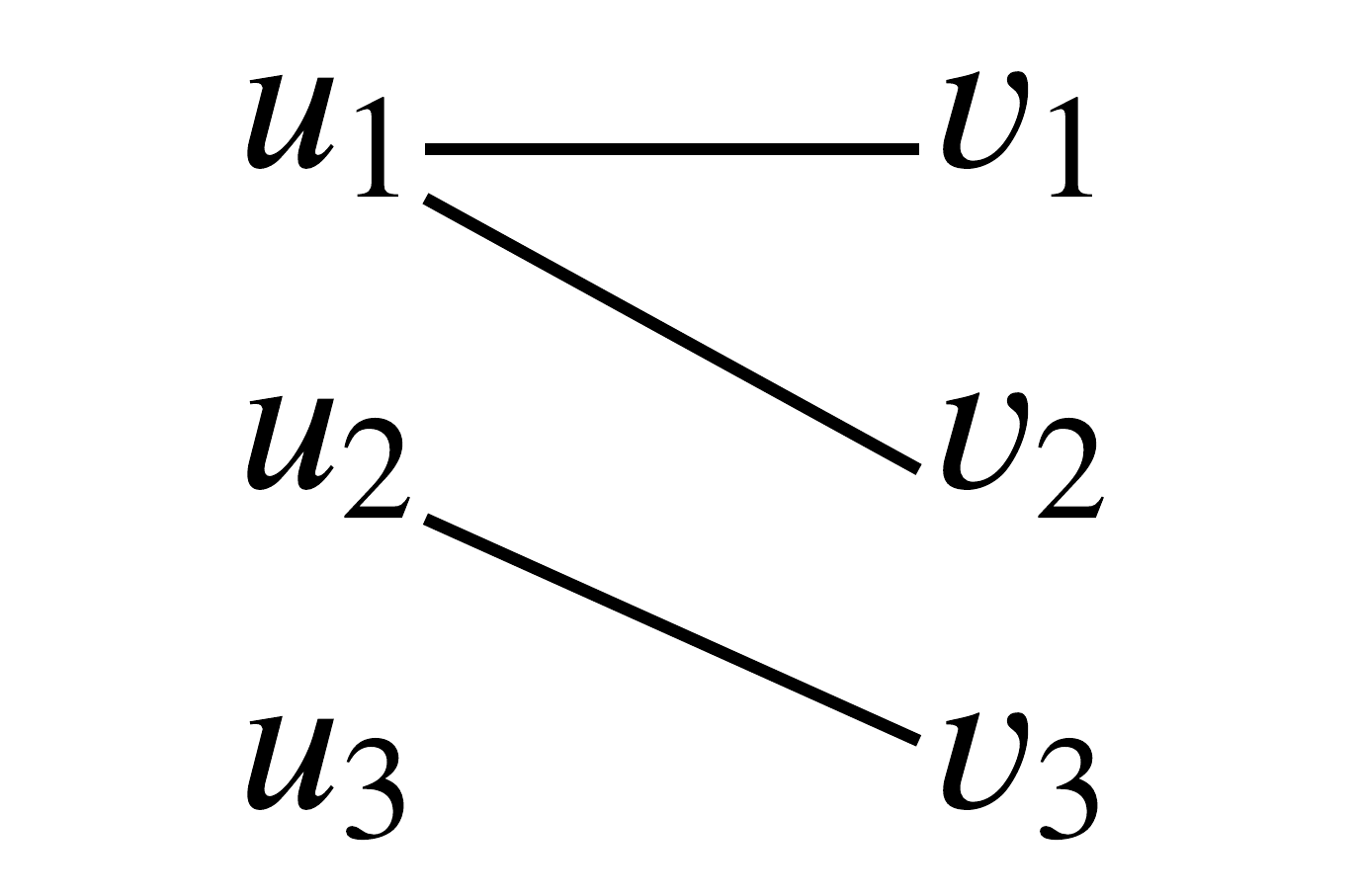}
        \vspace{-3mm}
    \end{subfigure}
    \begin{subfigure}[b]{0.42\textwidth}
        \includegraphics[width=\linewidth]{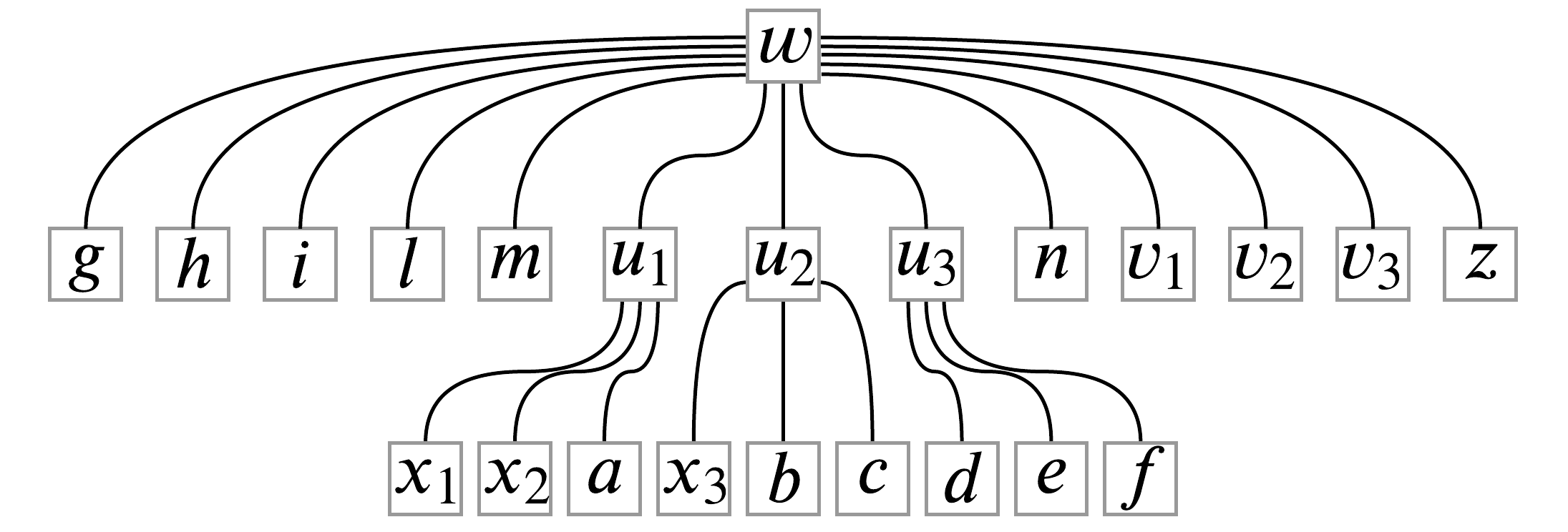}
        \vspace{-3mm}
    \end{subfigure}
    \begin{subfigure}[b]{0.42\textwidth}
        \includegraphics[width=\linewidth]{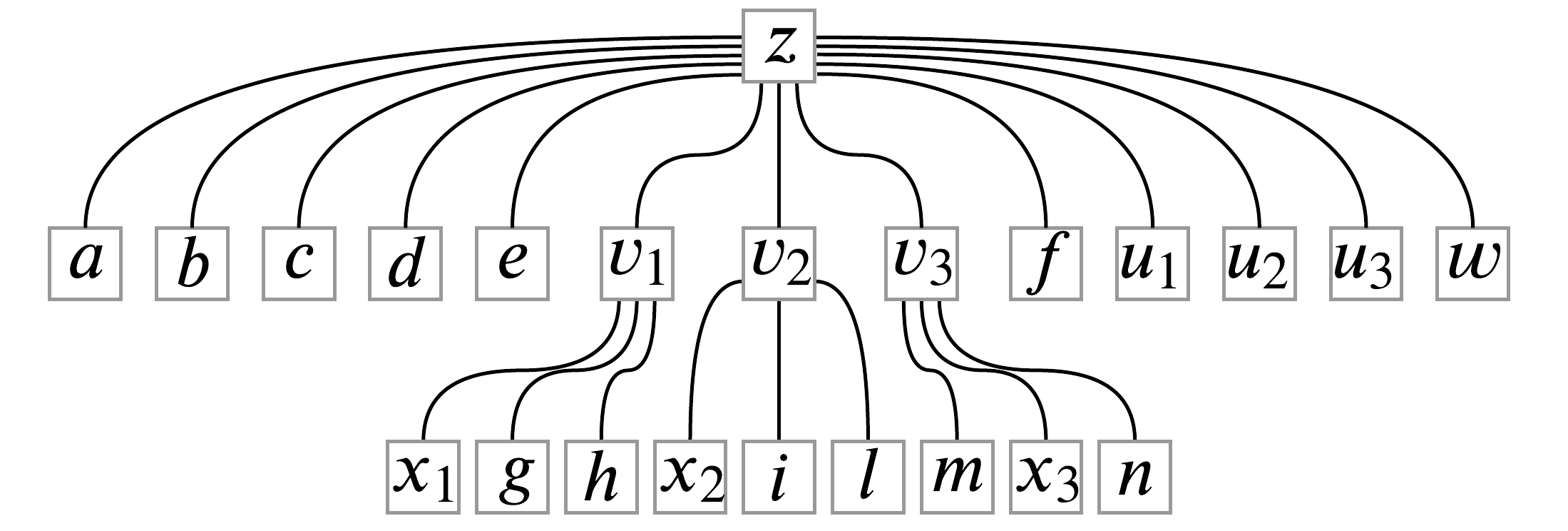}
        \vspace{-3mm}
    \end{subfigure}
    
    \caption{The two trees built for the graph on the left, according to Theorem~\ref{thm:frommatching}.}
    \label{fig:reduction}
\end{figure}

\section{A Constant-Factor Approximation Algorithm for the Rearrangement Distance}
\label{sec:approx}

A linear-time algorithm that, given two trees $\Tone$ and $\Ttwo$,
approximates $d(\Tone,\Ttwo)$ within a constant factor, was known for the case where at least one of the trees is binary~\cite{bernardini2019distance}: here we do not make any assumptions on the degrees.
Throughout this section, we actually consider $\tilde d(\Fone,\Ftwo)$, and show how to approximate it within
a constant factor. This allows us to approximate $d(\Tone,\Ttwo)$ within a constant factor
using the following procedure.
First, we add $n$ leaves $n+1,n+2,\ldots,n$ attached to the (identical) roots of $\Tone$ and $\Ttwo$ to obtain $\Tone'$ and $\Ttwo'$, respectively.
We call the resulting trees \emph{anchored}. Because $\Tone$ and $\Ttwo$ are assumed to have the
same root that cannot be permuted, we have $d(\Tone,\Ttwo)=d(\Tone',\Ttwo')$.
We claim that $\tilde d(\Tone',\Ttwo')=d(\Tone',\Ttwo')$.

\begin{lemma}
\label{lem:atmost}
For any two anchored trees $\Tone$ and $\Ttwo$, $\tilde{d}(\Tone,\Ttwo)= d(\Tone,\Ttwo)$.
\end{lemma}
\begin{proof}
Consider a sequence $s$ of link-and-cut and permutation operations
that transforms $\Tone$ into $\Ttwo$. We convert it into a sequence $s'$ of cut and permutation operations by simply
replacing every link-and-cut operation $v\,|\,u\rightarrow w$ with a cut operation $(v\dagger u)$. Let $\Tone'$ be the forest
obtained after applying $s'$ on $\Tone$. We claim that $\Tone' \sim \Ttwo$. Consider any $v\in [n]$.
Because we can assume that the permutation operation precedes all the link-and-cut operations in $s$,
if $p_{\Tone'}(v) \neq \bot$ then we must have $p_{\Ttwo}(v)=\bot$ or $p_{\Tone'}(v)=p_{\Ttwo}(v)$, as
$p_{\Tone'}(v)$ is the same as the parent of $v$ after applying $s$ on $\Tone$. This
shows that indeed $\Tone' \sim \Ttwo$, and so $\tilde{d}(\Tone,\Ttwo) \leq d(\Tone,\Ttwo)$.

For the other direction, we use the assumption that $\Tone$ and $\Ttwo$ are anchored trees on $2n$ nodes:
in both trees $r$ is the root and there are $n$ leaves $n+1,n+2,\ldots,n$ attached to $r$.
Observe that $\tilde d(\Tone,\Ttwo)<n$. We claim that an optimal sequence of cut and permutation operations 
doesn't permute $r$. Assume otherwise, then for every $u=n+1,n+2,\ldots,n$ either $u$
is also permuted, or we have a cut operation $(u \dagger r)$, so the size of the sequence must be at least $n$.
Now, let $s$ be an optimal sequence consisting of a permutation $\pi$ and then some cut operations,
and let $\Tone''$ be the tree obtained after applying $s$ on $\Tone$.
We obtain a sequence $s'$ of link-and-cut and permutation operations from $s$ as follows.
For every $v\in [n]$, if $p_{\Tone''}(v) = \bot$ and $p_{\Ttwo}(v)\neq \bot$, we locate the cut operation $(v\dagger u)$ in $s$ (there must be such operation, as $\Tone$ and $\Ttwo$ have the same root that is not permuted). In $s'$, we replace this operation with $v\,|\,u\rightarrow w$, where $w=p_{\Ttwo}(v)$.
Additionally, we reorder all link-and-cut operations to ensure that $w$ is not a descendant of $v$, which can be guaranteed by considering $v$ in the decreasing order of their levels in $\Ttwo$.
Let $\Tone'$ be the result of applying $s'$ on $\Tone$, and consider any $v\in [n]$.
If $p_{\Tone''}(v) \neq \bot$ and $p_{\Ttwo}(v) \neq \bot$ then $p_{\Tone'}(v)=p_{\Ttwo}(v)$
because $\Tone'' \sim \Ftwo$, and if $p_{\Tone''}(v) = \bot$ and $p_{\Ttwo}(v)\neq \bot$ then
$p_{\Tone'}(v)=p_{\Ttwo}(v)$ by the choice of $w$.
This shows that $s'$ transforms $\Tone$ into $\Ttwo$, thus $d(\Tone,\Ttwo) \leq \tilde{d}(\Tone,\Ttwo)$.
\end{proof}

\begin{figure}[t]
    \centering
   \includegraphics[width=.7\linewidth]{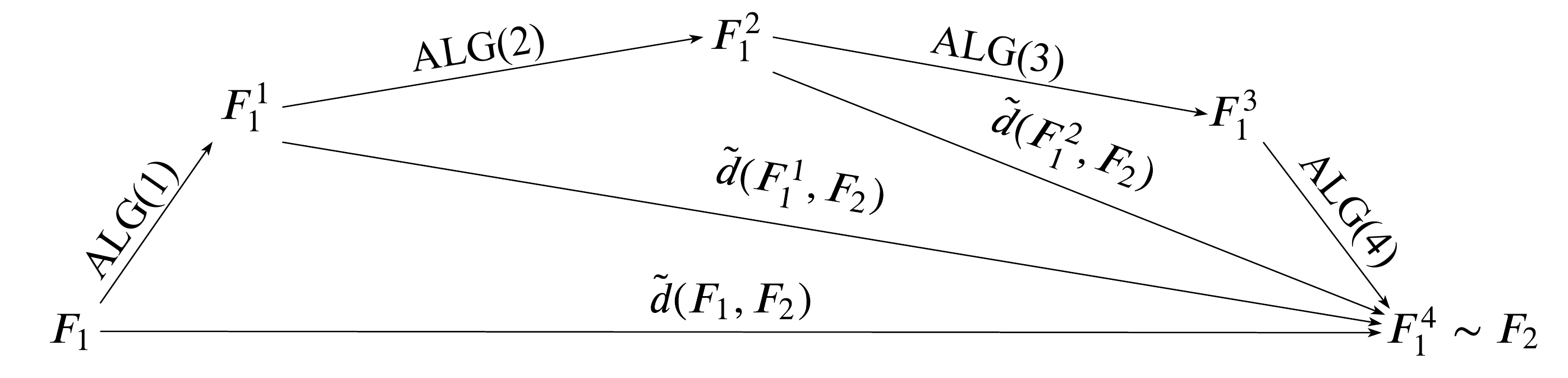}
   \caption{The four steps of the approximation algorithm.}
    \label{fig:approx}
\end{figure}

We can thus approximate $\tilde{d}(\Tone',\Ttwo')$ within a constant factor to obtain a constant factor
approximation of $d(\Tone,\Ttwo)$. In the remaining part of this section we design an approximation algorithm
for $\tilde{d}(\Fone,\Ftwo)$, where $\Fone$ and $\Ftwo$ are two arbitrary forests.

We start with describing the notation. Consider two forests $\Fone$ and $\Ftwo$.
For every $i\in [n]$, let $a[i]\in[n]$ be the parent of a non-root node $i$ in $\Fone$, and $a[i]=0$ if $i$ is a root in $\Fone$. Formally,
$a[i]=p_{\Fone}(i)$ when $p_{\Fone}(i)\neq \bot$ and $a[i]=0$ otherwise;
$b[i]$ is defined similarly but for $\Ftwo$.
We think of $a$ and $b$ as vectors of length $n$.

The algorithm consists of four steps, with step $j$ transforming forest $\Fone^{j-1}$ into $\Fone^{j}$ by performing
$\alg(j)$ operations, starting from $\Fone^{0}=\Fone$.
We will guarantee that $\alg(j) = \cO(\tilde{d}(\Fone^{j-1},\Ftwo))$.
Then, by triangle inequality and symmetry,
$\tilde{d}(\Fone^{j},\Ftwo)\leq \tilde{d}(\Fone^{j-1},\Fone^{j})+\tilde{d}(\Fone^{j-1},\Ftwo)\leq
\alg(j)+\tilde{d}(\Fone^{j-1},\Ftwo)=\cO(\tilde{d}(\Fone^{j-1},\Ftwo))$, so by induction $\tilde{d}(\Fone^{j},\Ftwo)=\cO(\tilde{d}(\Fone,\Ftwo))$.
Consequently, $\alg(j) = \cO(\tilde{d}(\Fone,\Ftwo))$, making the overall cost $\sum_{j}\alg(j) = \cO(\tilde{d}(\Fone,\Ftwo))$ 
In the $j$-th step of the algorithm $a[i]$ refers to the parent of $i$ in $\Fone^{j-1}$.
To analyse each step of the algorithm we will use the following two structures, the first of which is a streamlined version
of family partitions defined in the previous paper~\cite{bernardini2019distance}.

\begin{definition}[family partition]
\label{def:family}
Given two forests $\Fone$ and $\Ftwo$, the \emph{family partition}
$P(\Fone,\Ftwo)$ is the set $\{ (a[i],b[i]) : a[i],b[i]\neq 0~\land~a[i] \neq b[i] \}$.
\end{definition}

\begin{definition}[migrations graph]
\label{def:migrations}
Given two forests $\Fone$ and $\Ftwo$, the \emph{migrations graph}
$MG(\Fone,\Ftwo)$ consists of edges $ \{ (i,j) : a[i],a[j],b[i],b[j]\neq 0~\land~a[i]=a[j]~\land~b[i]\neq b[j] \} $.
\end{definition}

For a multiset $S$, let $|S|$ denote its cardinality, that is, the sum of multiplicities of all distinct elements of $S$.
The mode of $S$, denoted $\mode(S)$, is any element $s\in S$ with the largest multiplicity $\freq_{S}(s)$. We will use the following combinatorial lemma.

\begin{restatable}{lemma}{lempairs}
\label{lem:pairs}
Given any multiset $S$, let $f=\min\{|S|-\freq_{S}(\mode(S)),\lfloor |S|/2\rfloor\}$. All $|S|$ elements of $S$
can be partitioned into $f$ pairs $(x_{1},y_{1}),\ldots,(x_{f},y_{f})$,
$x_{i}\neq y_{i}$, for every $i\in [f]$, and the remaining $|S|-2f$ elements.
\end{restatable}

\begin{figure}[h]\hspace{-3mm}
    \begin{subfigure}[c]{0.52\textwidth}
       \includegraphics[width=1\textwidth]{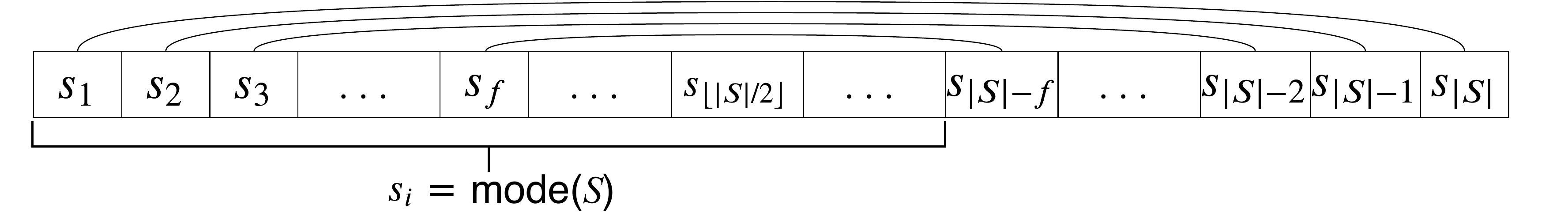}
    \end{subfigure}
    \hspace{-3mm}
    \begin{subfigure}[c]{0.52\textwidth}
       \includegraphics[width=1\textwidth]{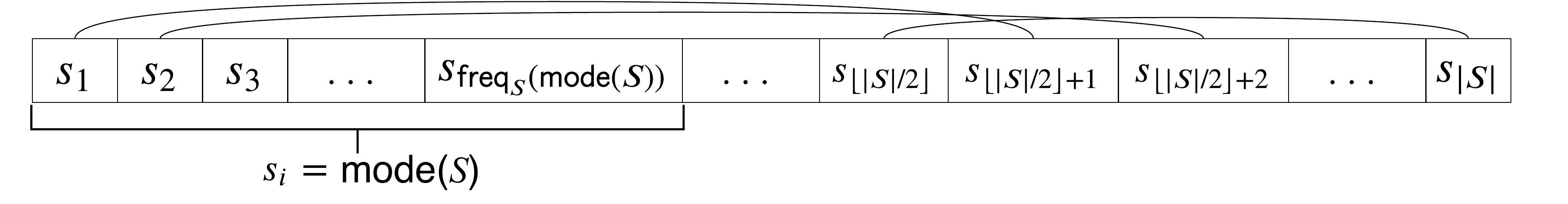}
    \end{subfigure}
    \caption{Pairing in the case $f=|S|-\freq_{S}(\mode(S))$ (left) and $f=\lfloor |S|/2\rfloor$ (right).}
    \label{fig:pairing}
\end{figure}

\begin{proof}
Number the elements of $S$ so that $s_1=\ldots =s_{\freq_S(\mode(S))}=\mode(S)$ and all of the others are sorted and numbered from $\freq_S(\mode(S))+1$ to $|S|$ accordingly.
Then, if $f=|S|-\freq_{S}(\mode(S))$, pairs $(s_i,s_{|S|-i+1})$, $i\in [f]$ are s.t. $s_i\neq s_{|S|-i+1}$ (Figure~\ref{fig:pairing}, left); if $f=\lfloor |S|/2\rfloor$, pairs $(s_i,s_{\lfloor |S|/2\rfloor+i})$, $i\in [\lfloor |S|/2\rfloor]$ are s.t. $s_i\neq s_{\lfloor |S|/2\rfloor+i}$ (Figure~\ref{fig:pairing}, right).
\end{proof}

\subsection{Step 1}
Roughly speaking, the aim of the first step is to ensure that all nodes that might be possibly involved in a permutation, i.e.,
the nodes with different children in $\Fone$ and $\Ftwo$, are roots. This is so that we do not need to worry about the relationship
with their parents.
For every $i\in [n]$ such that $a[i]$ and  $b[i]$ are both defined and different, we cut
the edges from $a[i]$ and $b[i]$ to their parents in $\Fone$, thus making both of them roots.
In other words, for every $i$ such that $a[i],b[i]\neq 0$ and $a[i]\neq b[i]$, we cut edges $(a[i],a[a[i]])$ and $(b[i],a[b[i]])$.
The resulting forest $\Fone^1$ has the following property: for each $i\in [n]$ such that
the parents of $i$ in $\Fone^{1}$ and in $\Ftwo$ are both defined and different,
$a[a[i]]=a[b[i]]=\bot$.

The number of cuts in this step is by definition at most twice the size of the family partition $P(\Fone,\Ftwo)$.
Bernardini et al.~\cite{bernardini2019distance} already showed that $|P(\Tone,\Ttwo)| \leq 2d(\Tone,\Ttwo)$ for two trees $\Tone$ and $\Ttwo$.
We show that this still holds for forests and $\tilde{d}$: for completeness, we provide a self-contained proof (cf. Lemma 16 in~\cite{bernardini2019distance}).
\begin{restatable}{lemma}{lempartition}
\label{lem:activeset}
$|P(\Fone,\Ftwo)|  \leq 2\tilde{d}(\Fone,\Ftwo) $, implying $\alg(1)\le 4\tilde{d}(\Fone,\Ftwo)$.
\end{restatable}

\begin{proof}
It is enough to verify that applying a single cut operation might decrease the size of the family partition by at
most one, while applying a permutation operation $\pi$
might decrease the size of the family partition by at most $2s$, where $s=|\{ u : u \neq \pi(u) \}|$.

Consider a cut operation $(v \dagger u)$. The only change to $a$ is that $a[v]$ becomes $0$,
so indeed the size of the family partition might decrease by at most one.

Now consider a permutation $\pi$. 
After applying $\pi$, an edge $(i,a[i])$ becomes $(\pi(i),\pi(a[i]))$, making $\pi(a[\pi^{-1}(i)])$ the parent of $i$.
This transforms the family partition $P$ into
\[ P'=\{ (\pi(a[i]), b[\pi(i)]) : a[i] \neq 0~\land~b[\pi(i)] \neq 0~\land~\pi(a[i]) \neq b[\pi(i)]\} .\]
To lower bound the size of $|P'|$, we first focus on the subset of $P$ corresponding to the nodes that are fixed by $\pi$. We therefore define
\[ P_{f} = \{ (a[i], b[i]) : a[i] \neq 0~\land~b[i] \neq 0~\land~a[i] \neq b[i]~\land~\pi(i)=i\} . \]
By definition, we can equivalently rewrite $P_{f}$ as
\[ P_{f} = \{ (a[i], b[\pi(i)]) : a[i] \neq 0~\land~b[\pi(i)] \neq 0~\land~a[i] \neq b[\pi(i)]~\land~\pi(i)=i\} . \]
Now consider all pairs with the same second coordinate $y$ in $P_{f}$: $(x_{1},y),(x_{2},y),\ldots,(x_{k},y)$, where $x_{i}\neq y$ for every $i\in [k]$.
$P'$ contains all pairs $(\pi(x_{i}),y)$ such that $\pi(x_{i})\neq y$. If $\pi(y)=y$ then $\pi(x_{i})= y$ cannot
happen and $P'$ contains all pairs with the second coordinate $y$ from $P_{f}$;
otherwise, $P'$ contains all such pairs except possibly one. Overall, $|P'| \geq |P_{f}|-s$, and $|P_{f}| \geq |P|-s$ so indeed $|P'| \geq |P|-2s$.
\end{proof}

\begin{figure}[t]
    \centering
   \includegraphics[width=0.8\linewidth]{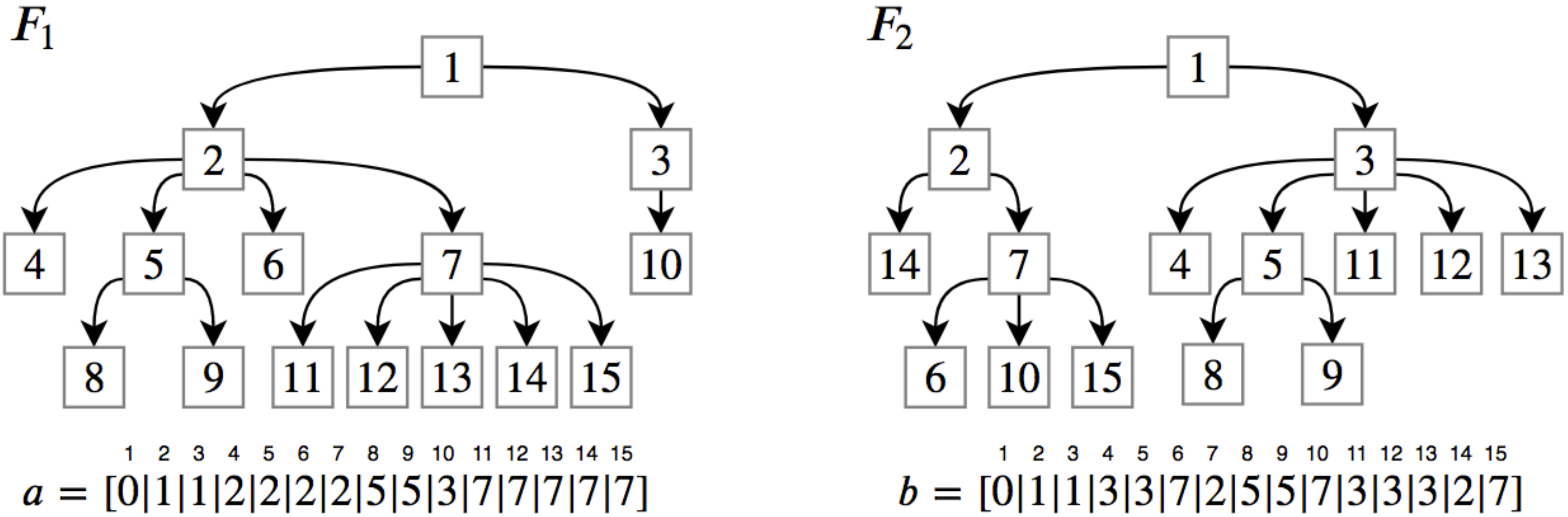}
    \caption{$\Fone$ and $\Ftwo$. The family partition is $P=\{(2,3), (2,7), (3,7), (7,3), (7,2)\}$.}
    \label{fig:trees}
\end{figure}

\begin{figure}\hspace{-1mm}
\centering
    \begin{subfigure}[b]{0.36\textwidth}
       \includegraphics[width=1\textwidth]{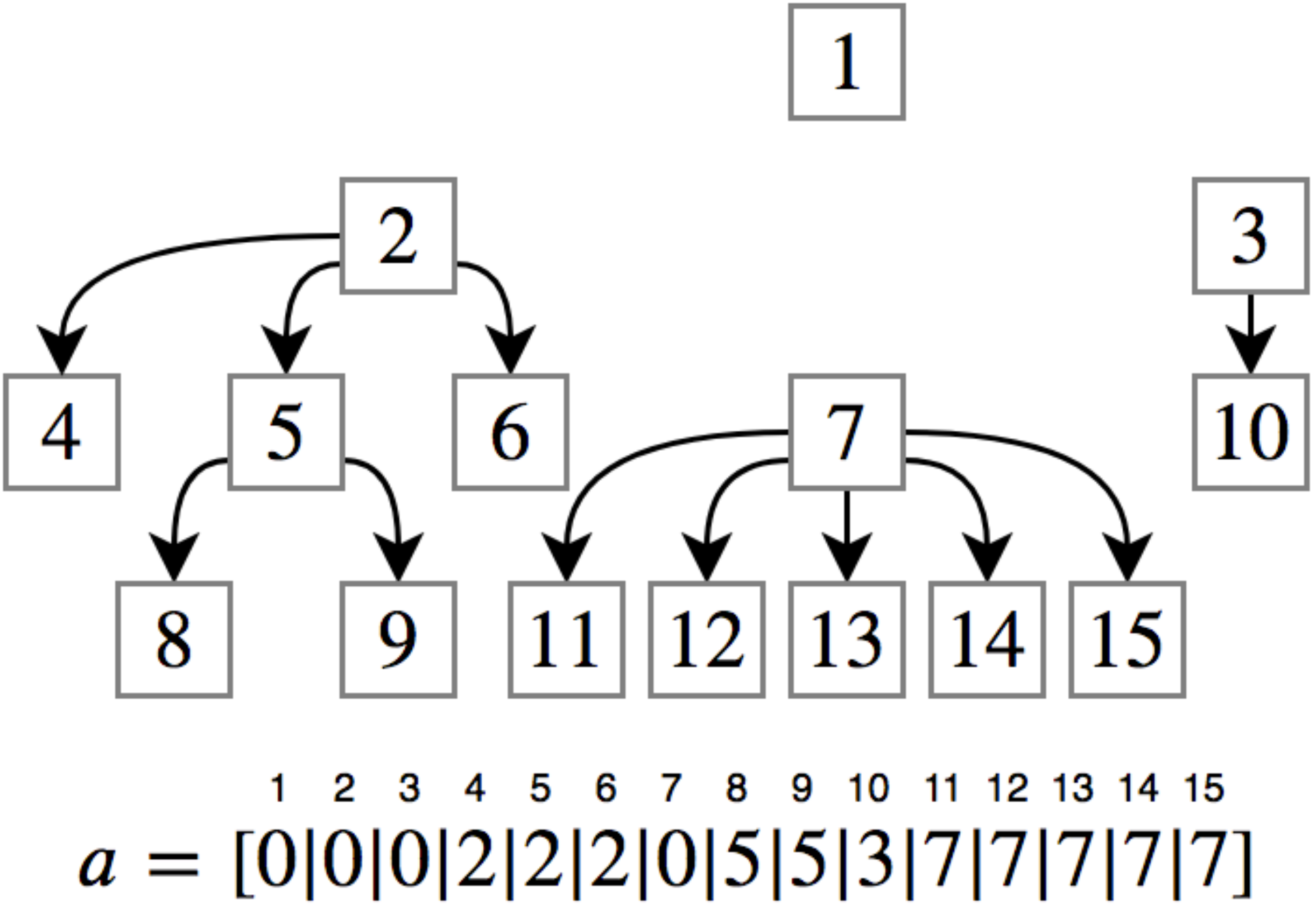}
       \vspace{-3mm}
        \caption{$\Fone^1$}
        \label{fone1}
    \end{subfigure}
    ~~~~ 
    \begin{subfigure}[b]{0.36\textwidth}
       \includegraphics[width=1\textwidth]{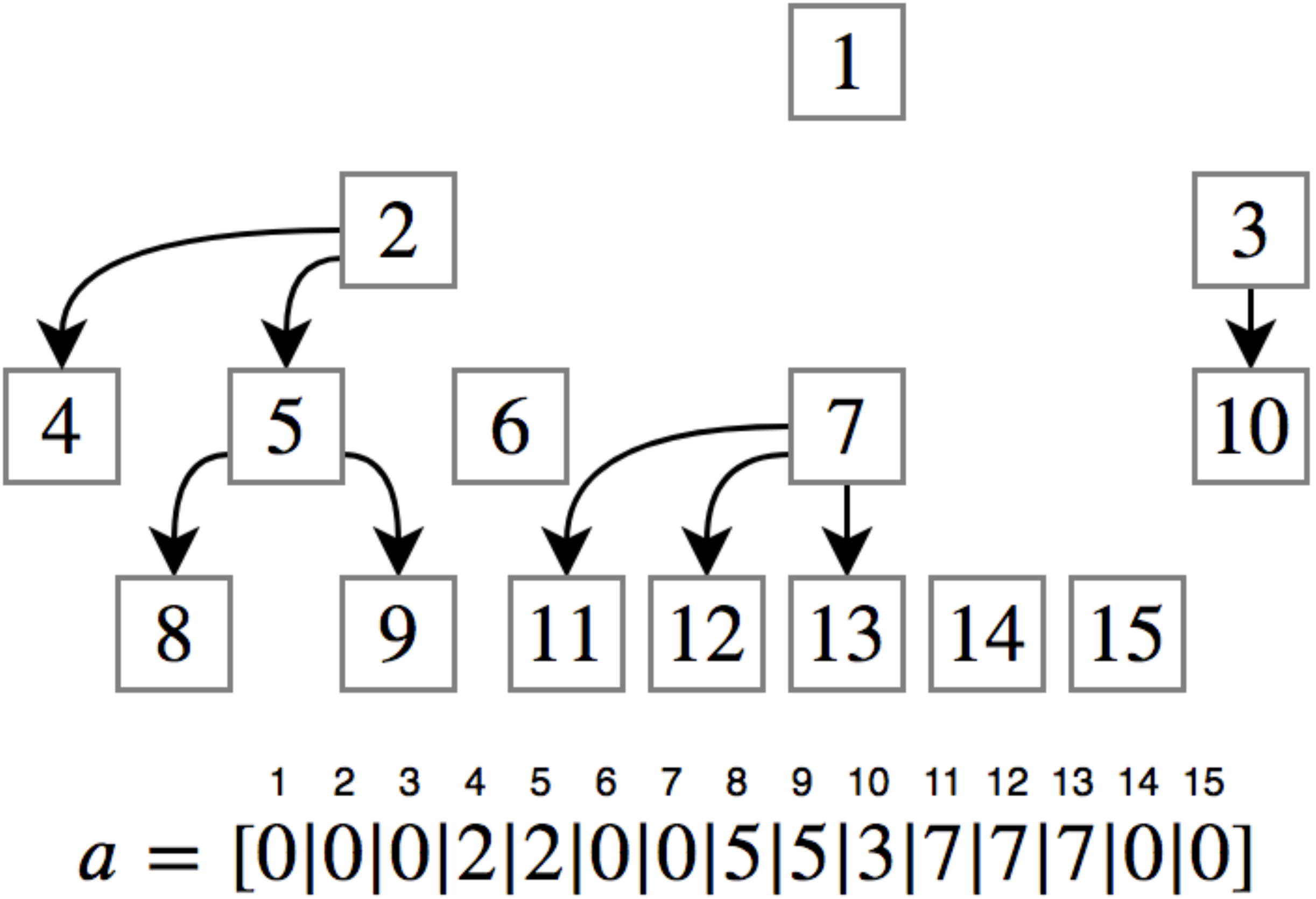}
       \vspace{-3mm}
        \caption{$\Fone^2$}
        \label{fone2}
    \end{subfigure}
  
    \begin{subfigure}[b]{0.36\textwidth}
       \includegraphics[width=1\textwidth]{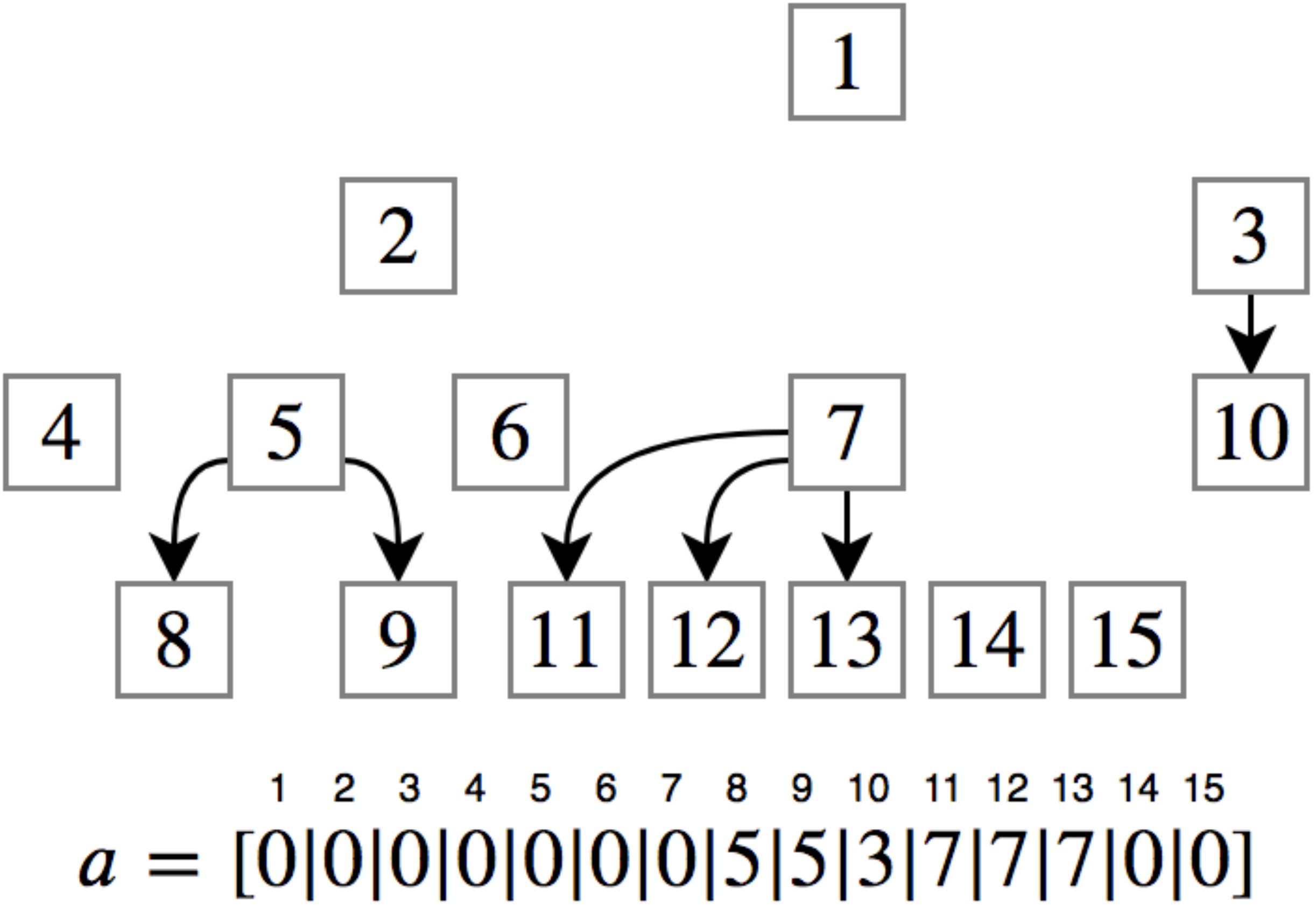}
      \vspace{-3mm}
        \caption{$\Fone^3$}
        \label{fone3}
    \end{subfigure}
    ~~~~
    \begin{subfigure}[b]{0.36\textwidth}
       \includegraphics[width=1\textwidth]{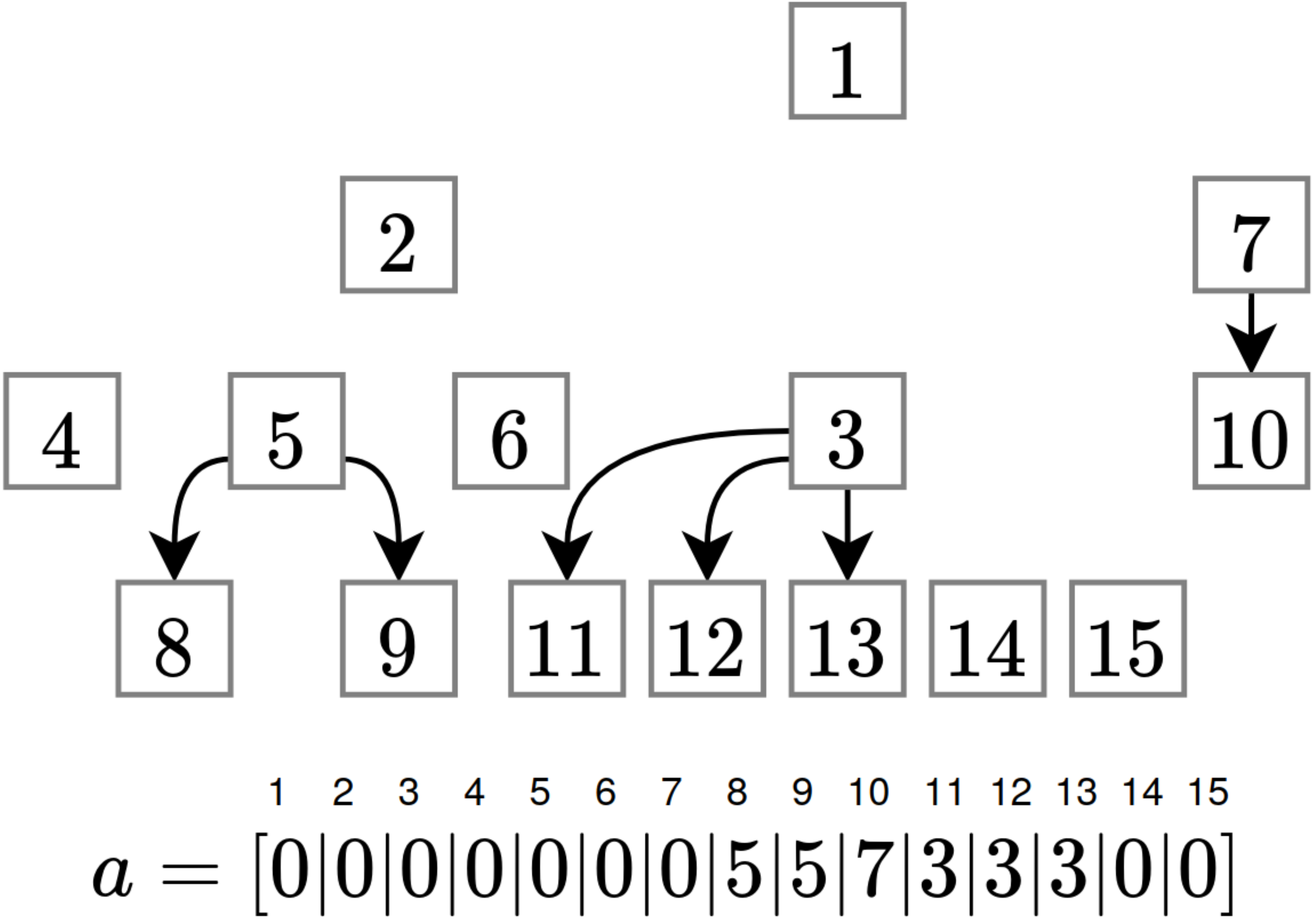}
      \vspace{-3mm}
        \caption{$\Fone^4$}
        \label{fone4}
    \end{subfigure}
    
    \caption{The forests obtained after Step 1 (\ref{fone1}), Step 2 (\ref{fone2}), Step 3 (\ref{fone3}) and Step 4 (\ref{fone4}).} \label{fig:forest}
\end{figure}

\begin{example}
\label{ex:step1}
Consider $\Fone$ and $\Ftwo$ depicted in Figure~\ref{fig:trees}. Step 1 consists of cut operations $(2\dagger1)$ (because, e.g., $a[4]\neq b[4]$ and $a[4]=2$), $(3\dagger1)$ (because $b[4]=3$) and $(7\dagger2)$ (because, e.g., $a[11]\neq b[11]$ and $a[11]=7$). The resulting forest $\Fone^1$ is shown in Figure~\ref{fone1}.
\end{example}

\subsection{Step 2}
Consider $u\in [n]$, and let $\children_{\Fone^{1}}(u)=\{v_{1},\ldots,v_{k}\}$.
We define the multiset $B(u)=\{b[v_{i}] : b[v_{i}]\neq 0 \}$ containing the parents in $\Ftwo$ of
the children of $u$ in $\Fone^{1}$.
Recall that $\mode(B(u))$ is the most frequent element of $B(u)$ (ties are broken arbitrarily).
We cut all edges $(v_{i},u)$ such that $b[v_{i}]\neq 0$ and $b[v_i]\neq \mode(B(u))$, and define, for each $u\in [n]$, its representative  $\rep(u)=\mode(B(u))$.
Intuitively, $\rep(u)$ is the node that might be convenient to replace $u$ with using a permutation. 
Roughly speaking, in this step we get rid of all of the children of $u$ that would be misplaced after permuting $u$ and $\rep(u)$, for each $u\in[n]$. The resulting forest $\Fone^2$ has the following property: for each $u\in [n]$, 
for any child $v$ of $u$ in $\Fone^{2}$, either $b[v]=0$ or $b[v]=\rep(u)$, i.e., the children of each node $u$ of $\Fone^2$ have all the same parent $\rep(u)$ in $\Ftwo$.

To bound the number of cuts in this step we first need a technical lemma relating the rearrangement distance of two
forests and the size of any matching in their migrations graph.

\begin{lemma}
\label{lem:opt}
Consider two forests $\Fone$ and $\Ftwo$ and their migrations graph $MG(\Fone,\Ftwo)$. For any matching $M$ in
$MG(\Fone,\Ftwo)$ it holds that $|M|\le \tilde{d}(\Fone,\Ftwo)$.
\end{lemma}

\begin{proof}
By definition, there is an edge between $i$ and $j$ in $MG(\Fone,\Ftwo)$ if and only if $a[i]=a[j]$, but $b[i]\neq b[j]$.
Let $M$ be any matching in $MG(\Fone,\Ftwo)$. If $|M| > 0$ then $\tilde{d}(\Fone,\Ftwo) \geq 1$,
so it is enough to show that, for a single operation transforming $\Fone$ into $\Fone'$,
the graph $MG(\Fone',\Ftwo)$ contains a matching $M'$ of size at least $|M|-s$, where $s=1$ for a cut operation and $s=|\{ u : u \neq \pi(u) \}|$ for a permutation operation $\pi$.

First, consider a cut operation $(v \dagger u)$. The only change in $MG(\Fone',\Ftwo)$ is removing
all edges incident to $v$. $M$ contains at most one edge incident to $v$, so we construct $M'$
of size at least $|M|-1$ from $M$ by possibly removing a single edge.
Second, consider a permutation operation $\pi$: we construct $M'$
from $M$ by removing every edge $(v,w)$ such that $v \neq \pi(v)$ or $w \neq \pi(w)$. Because
there is at most one edge incident to every $u$ such that $u \neq \pi(u)$, $M'$ contains
at least $|M|-s$ edges. $M'$ is a matching in $MG(\Fone',\Ftwo)$, as for every $(v,w)\in M'$
we have $p_{\Fone'}(v)=p_{\Fone}(v)$ and $p_{\Fone'}(w)=p_{\Fone}(w)$.
\end{proof}

\begin{lemma}
\label{lem:step2}
$\alg(2)\le 2\tilde{d}(\Fone^1,\Ftwo)$.
\end{lemma}

\begin{proof}
We consider each $u\in [n]$ separately. Let $m=\freq_{B_{u}}(\mode(B_u))$ and
$MG_{u}$ be the subgraph of $MG(\Fone^{1},\Ftwo)$ induced by $B_{u}$. We will first construct
a matching of appropriate size in every $MG_{u}$.
We cut every $(v_{i},u)$ such that $b[v_{i}]\neq 0$ and $b[v_{i}] \neq \mode(B_{u})$,
making $|B_{u}|-m$ cuts. Let $f=\min(|B_{u}|-m,\lfloor |B_{u}|/2\rfloor)$.
By Lemma~\ref{lem:pairs}, we can partition a subset of $B_{u}$
into $f$ pairs $(b[v_{i}],b[v_{j}])$ such that $b[v_i]\neq b[v_j]$. We add
every edge $(v_{i},v_{j})$ to the constructed matching. We claim that
$|B_{u}|-m \leq 2f$. This holds because $|B_{u}|-m \leq 2(|B_{u}|-m)$
and $|B_{u}|-m \leq |B_{u}|-1 \leq 2\lfloor |B_{u}/2|\rfloor$ for nonempty $B_{u}$.

We take the union of all such matchings to obtain a single matching $M$.
As argued above, the total number of cuts is at most $2|M|$.
Together with Lemma~\ref{lem:opt}, this implies that $\alg(2)\le 2|M| \le 2\tilde{d}(\Fone^1,\Ftwo)$.
\end{proof}

\begin{example}
\label{ex:step2}
Consider again $\Fone$ and $\Ftwo$ of Figure~\ref{fig:trees}.  $B(7)=\{3,3,3,2,7\}$, thus we cut $(14\dagger7)$ and $(15\dagger7)$. $B(2)=\{3,3,7\}$, implying  $(6\dagger2)$.
The resulting $\Fone^2$ is shown in Figure~\ref{fone2}.
\end{example}

\subsection{Step 3}
If after Step 2 all of the children of a node $u$ of $\Fone$ have the same parent $\rep(u)$ in $\Ftwo$, it still may be the case where $\rep(u)=\rep(v)$ with $u\neq v$, i.e., all of the children of two distinct nodes of $\Fone$ have the same parent in $\Ftwo$.
In this case, it is not clear how to choose whether to replace $u$ or $v$ with $\rep(u)=\rep(v)$ in a permutation.
This step aims at resolving this situation by cutting the ambiguous edges.

Consider thus $u\in [n]$, and let $\children_{\Ftwo}(u)=\{v_{1},v_{2},\ldots,v_{k}\}$. 
We define the multiset
$B'(u)=\{ a[v_{i}] : a[v_{i}] \neq 0 \}$ containing the parents in $\Fone^{2}$ of the children of $u$ in $\Ftwo$. We
cut all edges $(v_{i},a[v_{i}])$ such that $a[v_{i}]\neq 0$ and $a[v_{i}]\neq \mode(B'(u))$, breaking ties arbitrarily, and define $\rep'(u)=\mode(B'(u))$.
The resulting forest $\Fone^{3}$ has
the following property: for each $u\in [n]$, 
for any child $v$ of $u$ in $\Ftwo$, we have $a[v])=\bot$ or $a[v]=\rep'(u)$.

We observe that the number of cuts performed by the above procedure is the same as if we had applied
Step 2 on $\Ftwo$ and $\Fone^{2}$. Therefore, Lemma~\ref{lem:step2} implies the following.

\begin{lemma}
\label{lem:step3}
$\alg(3)\le 2\tilde{d}(\Fone^2,\Ftwo)$.
\end{lemma}

\begin{example}
\label{ex:step3}
Consider again $\Fone$ and $\Ftwo$ of Figure~\ref{fig:trees}. 
We have $B'(3)=\{2,2,7,7,7\}$, we thus cut $(4\dagger2)$ and $(5\dagger2)$.
The resulting forest $\Fone^3$ is shown in Figure~\ref{fone3}. 
\end{example}

\subsection{Step 4}

We summarize the properties of $\Fone^{3}$ and $\Ftwo$:
\begin{enumerate}
\item For each $u\in [n]$ such that $a[u],b[u]\neq 0$ and $a[u]\neq b[u]$, $a[u]$ and $b[u]$ are roots in $\Fone^{3}$.\label{prop:one}
\item For each $u\in [n]$ we can define $\rep(u)\in [n]$ in such a way that, for any child $v$ of $u$ in $\Fone^{3}$, we have $b[v]=0$ or $b[v]=\rep(u)$.\label{prop:two}
\item For each $u\in [n]$ we can define $\rep'(u)\in [n]$ in such a way that, for any child $v$ of $u$ in $\Ftwo$, we have $a[v]=0$ or $a[v]=\rep'(u)$.\label{prop:three}
\end{enumerate}
To finish the description of the algorithm, we show how to find a permutation operation $\pi$ of size $\cO(\tilde{d}(\Fone^{3},\Ftwo))$
that transforms $\Fone^{3}$ into $\Fone^{4}$ such that $\Fone^{4} \sim \Ftwo$.

For every $u$ such that $a[u],b[u]\neq 0$ and $a[u]\neq b[u]$, we require that $\pi(a[u])=b[u]$.
Due to Property~\ref{prop:one}, for every such $u$ we have ensured that $a[u]$ and $b[u]$ are roots
of $\Fone^{3}$. So, if we can find a permutation $\pi$ that satisfies all the requirements and does not perturb the non-roots
of $\Fone^{3}$, then it will transform $\Fone^{3}$ into $\Fone^{4}$ such that $\Fone^{3}\sim \Ftwo$.
Furthermore, if for every $x$ perturbed by $\pi$ there exists $u$ such that
$a[u],b[u]\neq 0$ and $a[u]\neq b[u]$ with $x=a[u]$ or $x=b[u]$ then by Lemma~\ref{lem:activeset}
$|\pi| \leq 2|P(\Fone^{3},\Ftwo)| \leq 4\tilde{d}(\Fone^{3},\Ftwo)$ as required.

To see that there indeed exists such $\pi$, observe that due to Property~\ref{prop:two}
there cannot be two requirements $\pi(x)=y$ and $\pi(x)=y'$ with $y\neq y'$.
Similarly, due to Property~\ref{prop:three} there cannot be two requirements $\pi(x)=y$ and $\pi(x')=y$ with $x\neq x'$. 
Thinking of the requirements as a graph, the in- and out-degree of every node is hence at most 1, so we can add extra edges to obtain a collection of cycles defining a permutation $\pi$ that does not perturb the nodes not participating in any requirement.

\begin{example}
Consider $\Fone$ and $\Ftwo$ of Figure~\ref{fig:trees}. $\pi=(3~7)$ transforms $\Fone^3$ into $\Fone^4\sim\Ftwo$.
The final $\Fone^{4}$ is shown in  Figure~\ref{fone4}.
\end{example}

\bibliographystyle{plainurl}
\bibliography{references}

\end{document}